\documentclass[11pt]{article}
\usepackage[a4paper, total={6.5in, 9.5in}]{geometry}

\usepackage{amsthm}
\usepackage{amsfonts}
\usepackage{amsmath}
\usepackage{amssymb}

\usepackage{stmaryrd}
\usepackage{color}
\usepackage{tikz}
\usetikzlibrary{positioning}
\usepackage{bm}
\usepackage{hyperref}
\usepackage{float}
\usepackage{ebproof}
\usepackage{multicol}
\usepackage{appendix}
\usepackage{comment}

\newcommand{\Nat}{\mathbb N}

\newtheorem{theorem}{Theorem}[section]

\newtheorem{lemma}[theorem]{Lemma}

\newtheorem{corollary}[theorem]{Corollary}

\newtheorem{definition}[theorem]{Definition}

\newcommand{\M}{\mathcal{M}}
\newcommand{\ModelM}{\M=(W,\sim,V)}
\newcommand{\lang}{\mathcal L}
\newcommand{\lPAL}{\lang_{PAL}}
\newcommand{\satisfies}{\vDash}

\newcommand{\eq}{\leftrightarrow}

\newcommand{\imp}{\rightarrow}
\newcommand{\Imp}{\Rightarrow}

\newcommand{\Dia}{\Diamond}

\renewcommand{\phi}{\varphi}

\newcommand{\inter}{\cap}

\newcommand{\K}{\square}
\newcommand{\update}{\sslash}

\newcommand{\sep}{\ \bm{\mid} \ }
\newcommand{\sepC}{ \bm{\mid}  }  
\newcommand{\calculus}{\textbf{DHS$_{PAL}$}}
\newcommand{\xMerge}{\otimes}

\newcommand{\byIH}{\overset{(IH)}{\rightsquigarrow} }
\newcommand{\byInvertibility}{\dashrightarrow}
\newcommand{\Rule}{\mathcal{R}}

\ebproofnewstyle{regular}{
left label template = \scriptsize\inserttext,
right label template = \small\inserttext}

\ebproofnewstyle{compact}{
left label template = \scriptsize\inserttext,
right label template = \small\inserttext,
separation=0.6em}

\ebproofnewstyle{inter}{
left label template = \scriptsize\inserttext,
right label template = \scriptsize\inserttext,
separation=0.8em}

\ebproofnewstyle{small}{
separation = 0.5em, rule margin = .5ex,
right label template = \scriptsize\inserttext,
template = \small$\inserttext$ }

\ebproofnewstyle{tiny}{
separation = 0.5em, rule margin = .5ex,
right label template = \scriptsize\inserttext,
template = \footnotesize$\inserttext$ }

\usepackage[backend=biber, style=trad-plain, sorting=nyt]{biblatex}

\begin{document}

\title{Dynamic Hypersequents for Public Announcement Logic}
\author{Clara Lerouvillois\footnote{IRIT, CNRS---INP---UT \& IHPST, CNRS---Paris 1 Panthéon Sorbonne. ORCID: 0009-0008-8826-9050. \\ Email: clara.lerouvillois@irit.fr}, Francesca Poggiolesi\footnote{IHPST, CNRS---Paris 1 Panthéon Sorbonne. ORCID: 0009-0007-4867-3799}}
\date{}

\maketitle

\begin{abstract}
Dynamic Epistemic Logic extends classical epistemic logic by modeling not only static knowledge but also its evolution through information updates. Among its various systems, Public Announcement Logic (PAL) provides one of the simplest and most studied frameworks for representing epistemic change, see \cite{DEL}. While the semantics of PAL is well understood as transformation of Kripke models, the existing proof theory might fail to fully capture this dynamism at the syntactic level.
In this paper we propose a step toward addressing this gap. In particular, building on the hypersequent calculus for S5 introduced in \cite{Poggiolesi2008}, we extend it with a mechanism that models the transition between epistemic models induced by public announcements. We call these structures \emph{dynamic hypersequents}. Using dynamic hypersequents, we construct a calculus for PAL and we show that it enjoys several desirable properties: admissibility of all structural rules (including contraction), invertibility of logical rules, as well as syntactic cut-elimination. \\


\end{abstract}

\section{Introduction}

Dynamic Epistemic Logic (DEL) \cite{baltagetal:1998,DEL} is a branch of logic that emerged in the late 1980s. It extends classical epistemic logic by modeling not only static states of knowledge but also the processes by which knowledge is revised. Today, the term DEL serves as an umbrella for a variety of extensions of epistemic logic equipped with dynamic operators. Among the simplest and historically most influential is the public announcement operator, which characterises one of the earliest and most studied forms of DEL: Public Announcement Logic (PAL) \cite{plaza1989}.

From both a Hilbert-style and a Kripke-semantic perspective, the different systems of dynamic epistemic logic display strong and distinctive features. Their most striking originality, however, lies in their semantics. Unlike in classical epistemic logic, which is concerned with static sets of possible worlds, dynamic epistemic logic focuses on transformations of these structures. After an announcement, for example, one typically moves from one model to another by altering the set of worlds or the accessibility relations between them.

The proof-theoretic side of DEL, and in particular PAL, presents a rather different picture. Several attempts have been made to develop effective proof systems for PAL. The first such attempt was the tableau method introduced in \cite{balbiani2007tableau, balbiani2010tab}, where results on soundness, completeness, and complexity were established. Later, sequent calculi were proposed. In \cite{nomura2015revising} and \cite{wu2023labelled}, two distinct labelled Gentzen-style calculi are developed. In the former, the dynamics of announcements are represented by relational atoms, while in the latter, rules are introduced that simulate reduction axioms -- reducing dynamic PAL-formulas into equivalent static ones. The first system admits a cut-elimination theorem, yielding a subformula property (modulo labels), whereas the second achieves only a weaker form of the subformula property. 
More recently, in \cite{liu2023}, a non-labelled sequent calculus was proposed, again using reduction-style rules akin to those in \cite{wu2023labelled}. In this calculus, only a pseudo-analytic cut-elimination theorem has been established.
Finally, it is worth emphasizing that algebraic approaches to dynamic epistemic logics have led to the development of proof systems that also display highly interesting structural properties, in particular for logics of information updates extending beyond PAL, see e.g. \cite{baltagandco,dyckhoffandco,frittellaandco}.

The main objective of this paper is to develop a sequent calculus for PAL that captures, by purely structural means, the dynamic character of dynamic epistemic logic, and of PAL in particular. More precisely, our approach builds on hypersequents (e.g. see \cite{avron,KULE,Restall}), and in particular on the hypersequent calculus for S5 introduced in \cite{Poggiolesi2008,Poggiolesi2010}. In this calculus, each sequent corresponds to a world in a model, and the collection of sequents represents the set of worlds in an S5 model, namely an S5 model. We extend this framework by introducing a mechanism that allows one to move from a hypersequent---representing a given model---to a new hypersequent, which corresponds to the updated model obtained through a public announcement. We refer to these new objects as \emph{dynamic hypersequents}. Dynamic hypersequents provide a way to represent a model together with its successive updates in a single object. They are therefore syntactically rich while, at the same time, avoiding the inclusion of explicit semantic elements. This combination gives rise to interesting structural properties, such as a genuine subformula property or admissibility of all structural rules (contraction rules included), while preserving the possibility of extending the framework to several different information updates. In this sense, dynamic hypersequents offer an original and expressive approach to syntactically capturing the dynamic nature of dynamic epistemic logic at the proof-theoretical level.

Within this framework, we construct a dynamic hypersequent calculus for PAL and establish several key results like the already mentioned admissibility of all structural rules, but also the invertibility of logical rules and syntactic cut-elimination. As far as we know, none of the calculi so far proposed for PAL satisfy all these features at the same time.

The paper is structured in the following way. In \emph{Section} \ref{sectionPAL} we will introduce the basic notions of PAL, whilst in \emph{Section} \ref{sectionDHS} we will introduce dynamic hypersequents as well as the calculus for PAL. \emph{Section} \ref{sectionAdmissibility} will be dedicated to the proof of the admissibility of the structural rules, and invertitbility of logical rules, and in \emph{Section} \ref{sectionSoundAndComplete} we will prove that our calculus is sound and complete with respect to PAL. Finally \emph{Section} \ref{sectionCut} will serve to prove syntactically the cut-elimination theorem. In \emph{Section} \ref{sectionConclusion} we will draw some conclusions and sketch several paths of future research.

\section{Public Announcement Logic}\label{sectionPAL}

This section provides a brief overview of the language, syntax, and semantics of PAL.
Although PAL is usually presented with a set of agents $\mathcal{A}$ and correspondingly as many epistemic modalities $K_a$ as agents $a\in \mathcal{A}$, in this work, for the sake of simplicity, we consider the single-agent version of PAL. The extension of our result to the multi-agent case is left for future research but sounds natural  in light of existing results for S5 with multiple agents \cite{Poggiolesi2013fromSingleToMany}.

\begin{definition}[Language $\lang_{PAL}$]
The language of public announcement logic $\mathcal{L}_{PAL}$ is composed of a countable set of atomic sentences,  $p, q,etc. $, the classical connectives $\neg$ and $\wedge$, the modal operator $\Box$, and the public announcement operator $[\cdot]$. Formulas of  $\lang_{PAL}$ are inductively defined as follows:

\begin{displaymath}
    A \ ::=\ p\ |\ \lnot A \ |\ (A \land A)\ |\ \K A \ |\ [A]A
\end{displaymath}
\end{definition}

We follow the standard rules for omission of the parentheses. The classical connectives $\vee, \rightarrow, \leftrightarrow$, as well as the dual modalities $\Dia$ and $\langle\cdot \rangle$ are defined as usual.

\begin{definition}[\textbf{PAL}]
    Based on the language $\mathcal{L}_{PAL}$, we introduce the Hilbert-style axiomatic system \textbf{PAL}, which is composed of the axioms and rules of inference in Table \ref{AxPAL}. As standard, we will write $\vdash_{\textbf{PAL}} A$, for the formula $A$ is derivable in the Hilbert system \textbf{PAL}.
\end{definition}

\begin{table}[h] \small
\centering
\begin{tabular}{|ll|}
\hline 
All instantiations of propositional tautologies &  \\
& \\
\textbf{Axioms of S5:} & \\
$ \K(A \imp B) \imp (\K A \imp \K B) $		& distributivity of $ \K $ (axiom K) \\ 
$ \K A \imp A $ 							& truth (axiom T) \\ 
$ \K A \imp \K\K A $ 						& positive introspection (axiom 4) \\ 
$ \lnot\K A \imp \K\lnot\K A $ 				& negative introspection (axiom 5) \\
& \\
\textbf{Reduction axioms:} & \\
$ [A]p \eq (A \imp p) $				 	    & atomic permanence \\ 
$ [A]\lnot B \eq (A \imp \lnot[A] B) $ 		& announcement and negation \\ 
$ [A](B \land C) \eq ([A]B \land [A]C) $ 	& announcement and conjunction \\ 
$ [A]\K B \eq (A \imp \K[A]B) $ 			& announcement and knowledge \\ 
$ [A][B]C \eq [A \land [A]B]C $ 			& announcement composition  \\ 
& \\
\textbf{Inference rules:} & \\
From $ A $ and $ A \imp B $, infer $ B $ 	& modus ponens (MP) \\ 
From $A$, infer $\K A$ 						& necessitation rule (Nec) \\ 
\hline
\end{tabular}
\caption{The axiomatisation \textbf{PAL}}\label{AxPAL}
\end{table}

We now present the semantics for PAL by first introducing the notion of standard epistemic model, and consequently the notions of satisfaction of a formula in a model and that of updated model.

\begin{definition}[Epistemic model]
An \emph{epistemic model} is a tuple $\ModelM$ where $W$ is a set of possible states, $\sim: W \rightarrow \mathcal{P}(W^2)$ is an accessibility relation and $V:P \rightarrow \mathcal{P}(W)$ is a valuation. For a model $\ModelM$ and a world $w \in W$, we call $(\M,w)$ a \emph{pointed epistemic model}, or simply a \emph{pointed model}.
\end{definition}

In this definition, $\sim$ is an equivalence relation, \emph{i.e.} it is reflexive, transitive, and symmetric.

\begin{definition}[Semantics of PAL]\label{defSemPAL}
Given an epistemic model $\ModelM$, the semantic relation $\satisfies$ and the updated model $M_A$ are simultaneously defined by induction as follows:
\begin{align*}
&\M,w \satisfies p                &   &\text{iff}  &   &w \in V(p) \\
&\M,w \satisfies \lnot A          &   &\text{iff}  &   &\M,w \nvDash A \\
&\M,w \satisfies A\land B         &   &\text{iff}  &   &\M,w \satisfies A \text{ and } \M,w \satisfies B \\ 
&\M,w \satisfies \K A             &   &\text{iff}  &   &\M,v \satisfies A \text{ for all } v \in W \text{ such that } w \sim v\\
&\M,w\satisfies[A]B & &\text{iff} & &\text{if }\M,w\satisfies A \text{ then } \M_A,w\satisfies B
\end{align*}

After announcement of $ A $, the updated model $ \M_A = ( W_{A},\sim_{A}, V_{A}) $ is defined as follows:
\begin{align*}
W_{A} &=  \lbrace w \in W : \M,w \satisfies A \rbrace \\
\sim_{A} &= \ \sim \inter\   W_{A} \\
V_{A}(p) &= V(p) \inter W_{A}
\end{align*}

A formula $A$ is \emph{valid}, denoted $\satisfies_{PAL} A$ if, and only if, for all models $\ModelM$ and all worlds $w \in W$, $\M,w\satisfies A$. The set of all validities is denoted PAL.
\end{definition}

\noindent Informally, to define the updated model $ \M_A $ we retain all and only the worlds in which $ A$ is true and leave the accessibility relation between remaining worlds and valuations at those worlds unchanged. In this sense, the effect of the announcement of $ A $ is to restrict the epistemic model to all and only the states where $A$ holds. Therefore, announcements can be seen as epistemic state \emph{updaters}.



For our purpose, we need to introduce a notation to denote a list of announcements. We do that in the following definition.

\begin{definition}[List of announcement]  We define a list of announcement $\alpha$ in a public announcement formula $[\alpha]B$ by induction in the following way:

\begin{itemize}
    \item $[\epsilon]B := B$
    
    \item $[\alpha \cdot A]B := [\alpha][A]B$
\end{itemize}

By consequence we have
\begin{align*}
    &\M_\alpha,w \satisfies [A]B & &\text{iff} & &\text{if } \M_\alpha,w \satisfies A \text{ then } \M_{\alpha \cdot A},w \satisfies B
\end{align*}
\end{definition}

\begin{theorem}[Soundness and Completeness of \textbf{PAL}]\label{AxComp} For all formulas $A \in \lPAL$, 
\begin{displaymath}
\vdash_{\textbf{PAL}} A \quad \text{ if, and only if, } \quad \satisfies_{PAL} A
\end{displaymath}
\end{theorem}

\begin{proof} The proof can be found in \cite{DEL}.
\end{proof}


\section{Dynamic hypersequents for PAL}
\label{sectionDHS}

In this section, we introduce \emph{dynamic hypersequents}, which extend classical hypersequents by incorporating a dynamic mechanism that enables transitions from one hypersequent to another. We introduce them step-by-step, starting from the standard notion of sequent.

\begin{definition}[Sequent]
A \emph{sequent} is a syntactic object of the form $M\Imp N$, where $M$ and $N$ are finite multisets of formulas. We denote sequents by capital greek letters $\Gamma, \Delta, ...$. By consequence, occurrences of specific formulas in sequents, e.g. $A, M\Rightarrow N$ or $M\Rightarrow N, A$, will be denoted by $A, \Gamma$ and $\Gamma, A$, respectively. Finally, when $\Gamma = M \Imp N$ and $\Delta = P \Imp Q$, the sequent $M,P \Imp N,Q$ is denoted $\Gamma\Delta$.
\end{definition}

\begin{definition}[Dynamic Sequent]
Let $\Gamma_1,\cdots,\Gamma_n$ be sequents and $\alpha_1,\cdots,\alpha_n$ be finite sequences of announcements pairwise distinct, \emph{i.e.} $\alpha_i\neq \alpha_j$ for all $1\leq i\neq j \leq n$. Then $\update_{\alpha_1}\Gamma_1 \update \cdots \update_{\alpha_n}\Gamma_n$ is a \emph{dynamic sequent}. For the empty sequence $\epsilon$ we write $\Gamma$ instead of $\update_\epsilon \Gamma$. In particular then $\Gamma$ is a dynamic sequent.
\end{definition}

Intuitively, a dynamic sequent represents a world $w$ in some model $\ModelM$, together with its copies in successively updated models $\M_\alpha$ that $w$ belongs to. For instance the dynamic sequent $M\Imp N \update_A M' \Imp N'$ represents the world $w$ in the initial model $\M$ as well as in the updated model $\M_A$. 

Note that in a dynamic sequent $X = \Gamma \update_{\alpha_1} \Gamma_1 \update \cdots \update_{\alpha_n} \Gamma_n$, we call $\update_\alpha \ \Gamma$ the $\alpha$-projection of $X$, or simply the \emph{$\alpha$-sequent}. Note that by definition, for any $\alpha$, the $\alpha$-projection of $X$ is unique. In particular, we call $\Gamma$ the \emph{initial sequent}.

\begin{definition}[Dynamic Hypersequent]
Let $X_1,\cdots,X_n$ be dynamic sequents. Then, $G := X_1 \sep \cdots \sep X_n$ is a \emph{dynamic hypersequent} (\emph{DHS} for short).
\end{definition}

Intuitively,  a dynamic hypersequent is the syntactic representation of a model $\ModelM$ and its successive updates $\M_\alpha$. For a better understanding, we can think of the set of updated models as a two-dimension table, where each row represents a world and its copies (if any) in the successive updated models, and each column represents an updated model, with the different worlds that belong to it. Dynamic hypersequents are a linear representation of such a table. As an example, the following table represents a model $\M$ with three worlds $w, u, v$, and three updated models obtained from $\M$ by announcing respectively $A$, $A$ and then $B$, and $C$.\footnote{In this table, an empty cell indicates that the corresponding world does not belong to the updated model.} 

\begin{center}
\begin{tabular}{|c|ccc|}
\hline 
$\M$ & $\M_A$ & $\M_{A\cdot B}$ & $\M_C$ \\ 
\hline
$w$ & $w'$ &  & $w''$ \\ 
$u$ &  &  & $u'$ \\ 
$v$ & $v'$ & $v''$ & $v'''$ \\ 
\hline 
\end{tabular} 
\end{center}

The  following dynamic hypersequent stands for the syntactic representation of the table above, where sequents $\Gamma,\Delta$ and $\Lambda$ represent worlds $w,u$ and $v$, respectively.
\begin{displaymath}
\Gamma \update_A \Gamma' \update_C \Gamma'' \sep \Delta \update_C \Delta' \sep \Lambda \update_A \Lambda' \update_{A \cdot B} \Lambda'' \update_C \Lambda'''
\end{displaymath}

In a hypersequent $G$, if we only consider initial sequents, then we get the syntactic representation of the initial model $\M$, called \emph{initial hypersequent}. Similarly, if we only consider $\alpha$-sequents, we obtain the $\alpha$ projection of $G$, that is a representation of the updated model $\M_\alpha$: this is called  \emph{$\alpha$-hypersequent}.


\begin{definition}[Interpretation]
The interpretation $\tau$ of a sequent $M \Imp N$ is the standard one, namely: $(M \Imp N)^\tau := \bigwedge M \imp \bigvee N$. We extend it to dynamic sequents as follows:
\begin{displaymath}
\update_{\alpha_1}\Gamma_1 \update \cdots \update_{\alpha_n} \Gamma_n := [\alpha_1]\Gamma_1^\tau \lor \cdots \lor [\alpha_n]\Gamma_n^\tau
\end{displaymath}
where, for all $1 \leq i \leq n$, $\alpha_i = A_{i_1}\cdot \cdots \cdot A_{i_{k_i}}$ and $[\alpha_i]B = [A_{i_1}]\cdots[A_{i_{k_i}}]B$, for $k_i \in \Nat$.

Finally, we extend the interpretation function $\tau$ to dynamic hypersequents as follows:
\begin{displaymath}
(X_1 \sep \cdots \sep X_n)^\tau := \K X_1^\tau \lor \cdots \lor \K X_n^\tau
\end{displaymath}
\end{definition}

We now have all ingredients to introduce the calculus $\calculus$, which is composed by the following axioms and inference rules.\\

\noindent \textbf{Axioms:}
\begin{flalign*}
    &\begin{prooftree}
    	\hypo{G \sepC X \update_\alpha p,M \Imp N,p}
    \end{prooftree} &
\end{flalign*}

\noindent \textbf{Propositional rules:}
\begin{flalign*}
    &\begin{prooftree}[regular]
		\hypo{ G \sepC X \update_\alpha M \Imp N, A }
		\infer1[(L$\lnot$)]{G\sepC X \update_\alpha \lnot A, M \Imp N}
	\end{prooftree} \hspace{6em}
	&
	&\begin{prooftree}[regular]
		\hypo{ G \sepC X \update_\alpha A, M \Imp N }
		\infer1[(R$\lnot$)]{G\sepC X \update_\alpha M \Imp N,\lnot A}
	\end{prooftree} & \\
	\hfill & \\
	&\begin{prooftree}[regular]
		\hypo{ G \sepC X  \update_\alpha A,B, M \Imp N }
		\infer1[(L$\land$)]{G\sepC X  \update_\alpha A\land B, M \Imp N}
	\end{prooftree} 
	& 
	&\begin{prooftree}[regular]
		\hypo{ G \sepC X \update_\alpha M \Imp N,A }
		\hypo{ G \sepC X \update_\alpha M \Imp N,B}
		\infer[separation = 0.8em]2[(R$\land$)]{G\sepC X \update_\alpha M \Imp N,A \land B}
	\end{prooftree} & \\
\end{flalign*}

\noindent \textbf{Modal rules:}
\begin{flalign*}
	&\begin{prooftree}[regular]
		\hypo{ G \sepC X \update_\alpha \K A,A, M \Imp N }
		\infer1[(L$\K_1$)]{G\sepC X \update_\alpha \K A, M \Imp N}
	\end{prooftree} 
	& 
	&\begin{prooftree}[regular]
		\hypo{ G \sepC X \update_\alpha M \Imp N \sep \update_\alpha \Imp A }
		\infer1[(R$\K$)]{G\sepC X \update_\alpha M \Imp N,\K A}
	\end{prooftree} & \\
	\hfill & \\
	&\begin{prooftree}[regular]
    	\hypo{ G \sepC X \update_\alpha \K A, M \Imp N \sepC Y \update_\alpha A,P \Imp Q}
    	\infer1[(L$\K_2$)]{ G \sepC X \update_\alpha \K A,M \Imp N \sepC Y \update_\alpha P \Imp Q}
	\end{prooftree} &
\end{flalign*}
\begin{flalign*}
	&\begin{prooftree}[regular]
    	\hypo{G \sepC \overline{X} \sep Y \update_{\alpha} \Delta, B_1}
    	\hypo{G \sepC \overline{X} \sep Y \update_\alpha \Delta \update_{\alpha \cdot B_1} \Imp B_2}
    	\hypo{\cdots}
    	\hypo{G \sepC \overline{X} \sep Y \update_\alpha \Delta \update_{\alpha\cdot\beta} A \Imp }
   		\infer4[(L$\K_3$)]{G \sepC X \update_{\alpha\cdot\beta} \K A, \Gamma \sep Y \update_\alpha \Delta}
	\end{prooftree} & \\
    \hfill\\
	&\text{where $\beta = B_1 \cdot B_2 \cdot \cdots \cdot B_n$ and $\overline{X}=X \update_{\alpha\cdot\beta} \K A, \Gamma$} & \\
\end{flalign*}

\noindent \textbf{Announcement rules:}
\begin{flalign*}
	&\begin{prooftree}[regular]
    	\hypo{G \sepC X \update_\alpha M \Imp N,A \update_{\alpha \cdot A} M' \Imp N'}
    	\hypo{G \sepC X \update_\alpha M \Imp N \update_{\alpha\cdot A} B,M' \Imp N'}
    	\infer2[(L$[\cdot]$)]{ G \sepC X \update_\alpha [A]B, M \Imp N \update_{\alpha\cdot A} M' \Imp N' } 
	\end{prooftree} \hfill & \\
	\hfill & \\
	&\begin{prooftree}[regular]
    	\hypo{G \sepC X \update_\alpha A, M \Imp N \update_{\alpha\cdot A} M' \Imp N', B}
    	\infer1[(R$[\cdot]$)]{ G \sepC X \update_\alpha M \Imp N, [A]B \update_{\alpha\cdot A} M' \Imp N' }
	\end{prooftree} & \\
\end{flalign*}

\noindent \textbf{Dynamic rules:}
\begin{flalign*}
	&\begin{prooftree}[regular]
    	\hypo{G \sepC X \update_\alpha p,M \Imp N \update_{\alpha \cdot A} p, M' \Imp N'}
    	\infer1[(Lat)]{G \sepC X \update_\alpha M \Imp N \update_{\alpha \cdot A} p,M'\Imp N'}
	\end{prooftree}
	&
	&\begin{prooftree}[regular]
    	\hypo{G \sepC X \update_\alpha M \Imp N, p \update_{\alpha \cdot A} M' \Imp N', p}
    	\infer1[(Rat)]{G \sepC X \update_\alpha M \Imp N \update_{\alpha \cdot A} M' \Imp N',p}
	\end{prooftree} & \\
	\hfill & \\
	&\begin{prooftree}[regular]
    	\hypo{G \sepC X \update_\alpha \Imp \update_{\alpha \cdot A} \, M \Imp N}
    	\infer1[(New)]{G \sepC X \update_{\alpha \cdot A} M \Imp N}
	\end{prooftree}
	&
	&\begin{prooftree}[regular]
    	\hypo{G \sepC X \update_\alpha A,M \Imp N \update_{\alpha \cdot A} M' \Imp N'}
    	\infer1[(Recall)]{G \sepC X \update_\alpha M \Imp N \update_{\alpha \cdot A} M' \Imp N'}
	\end{prooftree} &\\
\end{flalign*}

In both the announcement rules (L$[\cdot]$) and (R$[\cdot]$), if the dynamic sequent $\update_{\alpha\cdot A}\ M' \Imp N'$ does not occur in the conclusion, we add it to their premises with $M'$ and $N'$ instantiated as empty multisets. So, for example, the rule (L$[\cdot]$) might have the following form:
\begin{center}
	\begin{prooftree}[regular]
    	\hypo{G \sepC X \update_\alpha M \Imp N,A }
    	\hypo{G \sepC X \update_\alpha M \Imp N \update_{\alpha\cdot A} B\Imp}
    	\infer2[(L$[\cdot]$)]{ G \sepC X \update_\alpha [A]B, M \Imp N} 
	\end{prooftree}
\end{center}

\noindent In the rest of the paper, (L$[\cdot])_1$ and (R$[\cdot])_1$ will denote instances of the rules where the dynamic sequent $\update_{\alpha\cdot A}\ M' \Imp N'$ occurs in the conclusion;   (L$[\cdot])_2$ (R$[\cdot])_2$ will denote instances of the rules where the dynamic sequent $\update_{\alpha\cdot A}M' \Imp N'$ does not occur in the conclusion. 

    

We now turn to a more informal discussion of the rules of the calculus. The propositional rules are the standard ones from classical logic, adapted to the dynamic hypersequent framework. The modal rules---(L$\K_1$), (L$\K_2$), and (R$\K$)---are essentially those introduced in \cite{Poggiolesi2008}, though once again reformulated to fit within the dynamic hypersequent setting.

The announcement rules become intuitive when viewed through the Kripke semantics satisfability relation of public announcement formulas. Consider the rules read bottom-up. For (R$[\cdot]$), let $\ModelM$ be a model and $w$ be a world of that model represented by the dynamic sequent $X$. If $[A]B$ is false at $(\M_\alpha,w)$, then $A$ must be true at $w$ (so that it can indeed be announced): this is displayed via the sequent $\update_\alpha \ A, M \Imp N$. After this announcement, however, $B$ is false at $w$ in the further updated model $\M_{\alpha \cdot A}$, which is displayed by the sequent $\update_{\alpha \cdot A} \ M'\Imp N',B$.

Similarly, for (L$[\cdot]$): if $[A]B$ is true at $w$ in $\M_\alpha$, then one of the following two options must hold: either $A$ is false at $(\M_\alpha,w)$, in which case the announcement cannot be made, or else $A$ is true at $(\M_\alpha,w)$, so that $w \in \M_{\alpha \cdot A}$, and $B$ holds at $(\M_{\alpha \cdot A},w)$. The former case is represented in the left-premise by $\update_\alpha\ M \Imp N, A$ while the latter is displayed in the right-premise by $\update_{\alpha \cdot A}\ B, M' \Imp N'$.

The dynamic rules capture important structural properties of models and their updates. Rules (Lat) and (Rat) express atomic permanence: read bottom-up, (Lat) for example states that if $p$ holds at $w$ in the updated model $\M_{\alpha \cdot A}$---represented by $\update_{\alpha \cdot A}\ p,M' \Imp N'$---then $p$ must also hold at $w$ in the previous model $\M_\alpha$---which corresponds to $\update_{\alpha}\ p,M \Imp N$. Rule (New) ensures that if $w \in \M_{\alpha \cdot A}$, then $w$ must already belong to the previous model $\M_\alpha$: hence from $\update_{\alpha \cdot A}\ M \Imp N$ we get $\update_\alpha \Imp$. Finally, (Recall) requires that if $w \in \M_{\alpha \cdot A}$---displayed by $\update_{\alpha \cdot A}\ M' \Imp N'$ ---, then $A$ must be true at $(\M_\alpha,w)$, which corresponds to $\update_\alpha\ A, M \Imp N$.

Finally, rule (L$\Box_3$) should be seen as a dynamic modal rule, since it displays the relationship between knowledge and announcements. It can be understood as follows: for $\K A$ to be true at some world $w$ in some updated model $\M_{\alpha \cdot \beta}$, every accessible world in $\M_{\alpha \cdot \beta}$ must satisfy $A$. Hence, for any other world $v\in \M_\alpha$, either $v\notin \M_{\alpha \cdot \beta}$, or $A$ is true at $v$ in $\M_{\alpha \cdot \beta}$. The latter case is represented by the rightmost premise of the rule. The former case, on the other hand, is represented by all the remaining premises: indeed, for $v \in \M_\alpha$ not to belong to $\M_{\alpha \cdot \beta}$, one of the following must hold (where $\beta = B_1 \cdots B_n$): $B_1$ is false at $(\M_\alpha,v)$, so $v$ does not belong to $\M_{\alpha\cdot B_1}$, or $v \in \M_{\alpha\cdot B_1}$ but $B_2$ is false at $(\M_{\alpha\cdot B_2},w)$, etc. This is why (L$\K_3$) has $n+1$ premises, when the sequence $\beta$ is composed of $n$ announcements.

 As an example, consider the following instance of the rule, where $\beta = B$ and $G$ is empty:

\begin{center}
\begin{prooftree}[regular]
    \hypo{ X \update_{\alpha\cdot B} \K A, \Gamma \sep Y \update_\alpha \Delta, B}
    \hypo{ X \update_{\alpha\cdot B} \K A, \Gamma \sep Y \update_\alpha \Delta \update_{\alpha \cdot B} A \Imp }
    \infer2[(L$\K_3$)]{ X \update_{\alpha\cdot B} \K A, \Gamma \sep Y \update_\alpha \Delta }
\end{prooftree}
\end{center}
Let $\ModelM$ be a model where $W = \lbrace w,v \rbrace$. Worlds $w, v$ are represented by $X \update_{\alpha\cdot B} \K A, \Gamma$ and $Y \update_{\alpha}  \Delta$ respectively. Then, we have $w \in W_{\alpha\cdot B}$ (which implies $w \in W_\alpha$) and $v \in W_{\alpha}$. For $\M_{\alpha \cdot B},w \satisfies \K A$ to hold, either $v \notin W_{\alpha \cdot B}$ or $\M_{\alpha \cdot B}, v \satisfies A$. The first case arises when $\M_\alpha, v \nvDash B$, which corresponds, in the first premise, to the sequent $\update_\alpha \ \Delta,B$. The second case is captured, in the second premise, by the sequent $\update_{\alpha\cdot B} \, A \Imp$.

\begin{definition} \emph{Derivations} in $\calculus$ are defined in the standard way and are denoted by $d, d', etc.$. The \emph{height} of a derivation $d$, denoted $h(d)$, is also inductively defined in the standard way (see \cite{BasicProofTheoryTroelstra}). As usual, we write $\vdash_{\calculus} G$ to denote that the dynamic hypersequent $G$ is derivable in the calculus \textbf{DHS}$_{PAL}$.
\end{definition}


\section{Admissibility of the structural rules and invertibility of the logical rules}\label{sectionAdmissibility}

In this section, we show that the structural rules of weakening, contraction and merge are hp-admissible, and that all logical rules are hp-invertible. We also introduce new dynamic structural rules and we prove their hp-admissibility.

First of all, we introduce the following notion that will prove useful to define the structural rules.
\begin{definition}
    Let $X = \update_{\alpha_1} \, \Gamma_1 \update \cdots \update_{\alpha_n} \Gamma_n$ be a dynamic sequent. The labels of $X$ are all the sequences of announcements that label some sequent in $X$, \emph{i.e.} $labels(X)= \lbrace \alpha_1, \ldots, \alpha_n \rbrace$. 
\end{definition}


\begin{lemma}\label{lemmaDerivabilityTrivialSequents}
All dynamic hypersequents of the form $G \sep X \update_\alpha A, M \Imp N, A$ are derivable.
\end{lemma}

\begin{proof} By standard induction on the complexity of $A$.
\end{proof}



\begin{lemma}\label{lemmaAdmissibilityIW}
The rule of \emph{internal weakening} is hp-admissible.
\begin{displaymath}
\begin{prooftree}[regular]
    \hypo{G \sep X \update_\alpha M \Imp N}
    \infer1[ (IW)]{G \sep X \update_\alpha P,M \Imp N,Q}
\end{prooftree}
\end{displaymath}
\end{lemma}

\begin{proof}
    The proof proceeds by straightforward induction on the height of the derivation of the premise. 
\end{proof}

\begin{lemma}\label{lemmaAdmissibilityEW}
The rule of \emph{external weakening} is hp-admissible.
\begin{displaymath}
\begin{prooftree}[regular]
    \hypo{G}
    \infer1[(EW)]{G \sep X}
\end{prooftree}
\end{displaymath}
\end{lemma}

\begin{proof}
    The proof proceeds by induction on the height of the derivation of the premise. If $G$ is an axiom, so is $G \sep X$. If $G$ is not an axiom, we distinguish cases on the last applied rule $ \Rule$: in each of the cases, we apply the inductive hypothesis to the premise(s) of $ \Rule$ and then use $ \Rule$ again. We only show one example, whilst the others can be treated similarly.

    If the last applied rule is (R$[\cdot])_2$, we proceed as follows:
    \begin{center}
        \begin{prooftree}[regular]
            \hypo{}
            \ellipsis{$d$}{  G \sep Y \update_\alpha A, \Gamma \update_{\alpha \cdot A} \Imp B}
            \infer1[(R$[\cdot])_2$]{  G \sep Y \update_\alpha \Gamma, [A]B}
        \end{prooftree}
        $\byIH$\footnote{From now on, by the symbol $\byIH$, we will denote the fact that the premise(s) of the inference on the right is obtained by inductive hypothesis on the premise(s) of the inference on the left.} \
        \begin{prooftree}[regular]
            \hypo{}
            \ellipsis{$d'$}{  G \sep Y \update_\alpha A, \Gamma \update_{\alpha \cdot A} \Imp B \sep X}
            \infer1[(R$[\cdot])_2$]{  G \sep Y \update_\alpha \Gamma, [A]B \sep X}
        \end{prooftree}
    \end{center}
\end{proof}

\begin{lemma}\label{lemmaAdmissibilityDW}
The rule of \emph{dynamic weakening} is hp-admissible, where $\alpha \notin labels(X)$.
\begin{displaymath}
    \begin{prooftree}[regular]
        \hypo{G \sep X}
        \infer1[(DW)]{G \sep X \update_\alpha M \Imp N}
    \end{prooftree}
\end{displaymath}
\end{lemma}

\begin{proof}
The proof proceeds by induction on the height of the derivation of the premise and is similar to the one of the above lemma.
Note that in the inductive step, we may need to use hp-admissibility of (IW) and replace some instances of (L$[\cdot])_2$ (resp. (R$[\cdot])_2$) by instances of (L$[\cdot])_1$ (resp. (R$[\cdot])_1$). 
We provide one example below.

\begin{center}
    \begin{prooftree}[regular]
        \hypo{G \sep X \update_\alpha A, M \Imp N \ \update_{\alpha \cdot A} \Imp B}
        \infer1[(R$[\cdot])_2$]{G \sep X \update_\alpha M \Imp N, [A]B}
    \end{prooftree}
\end{center}

\noindent Here we use the hp-admissibility of (IW) and apply (R$[\cdot])_2$ to obtain the desired conclusion.
\begin{center}
    \begin{prooftree}[regular]
        \hypo{G \sep X \update_\alpha A, M \Imp N \ \update_{\alpha \cdot A} M' \Imp N', B}
        \infer1[(R$[\cdot])_1$]{G \sep X \update_\alpha M \Imp N, [A]B \ \update_{\alpha \cdot A} M' \Imp N'}
    \end{prooftree}
\end{center}

\end{proof}


\begin{lemma}\label{lemmaInvertible}
    All logical rules are hp-invertible.
\end{lemma}

\begin{proof}
    For each rule, the proof proceeds by induction on the height of a derivation of the conclusion. \\

    \noindent$-$ Propositional rules: we only show invertibility of (L$\land$), because the method is the same for the others. If $G \sep X \update_\alpha A \land B, M \Imp N$ is an axiom, then so is $G\sep X \update_\alpha A,B, M \Imp N$. Otherwise, we consider the last occurrence of a rule $\Rule$ in the derivation. If $ \Rule$ has $A\land B$ as its principal formula, then it must be an occurrence of (L$\land)$ and the premise is precisely $G \sep X \update_\alpha A,B,M \Imp N$. If $A\land B$ is not principal in $\Rule$, we apply the inductive hypothesis to the premise(s) of $\Rule$ and then use $\Rule$ again. \\

    \noindent$-$ Modal rules: (L$\K_1$) and (L$\K_2$) are hp-invertible by hp-admissibility of (IW) and (L$\K_3$) is hp-invertible by hp-admissibility of (IW) and (DW). We now dwell on the hp-invertibility of (R$\K$). If $G \sepC X \update_\alpha M \Imp N, \K A$ is an axiom, then so is $G \sepC X \update_\alpha M \Imp N \sep \update_\alpha \Imp A$. Otherwise, we consider the last occurrence of a rule $\Rule$ in the derivation. If the principal formula of $\Rule$ is $\K A$, then $\Rule$ is the rule (R$\K$) and its premise is precisely what we want, namely $G \sepC X \update_\alpha M \Imp N \sep \update_\alpha \Imp A$. If $\K A$ is not principal in $\Rule$, we apply the inductive hypothesis to the premise(s) of the rule and then use $\Rule$ again. We show one example where $\Rule$ is an occurrence of (L$\K_3$).

    \begin{center}
    \begin{prooftree}[inter]
        \hypo{}
        \ellipsis{$d_1$}{G \sepC X \update_\alpha \Gamma, \K A, B_1 \sepC Y \update_{\alpha\cdot \beta} \K C,\Delta}
        \hypo{\cdots}
        \hypo{}
        \ellipsis{$d_n$}{G \sepC X \update_\alpha \Gamma, \K A \update_{\alpha \cdot \beta} C \Imp \sepC Y \update_{\alpha\cdot \beta} \K C,\Delta}
        \infer3[(L$\K_3$)]{G \sepC X \update_\alpha \Gamma, \K A \sepC Y \update_{\alpha\cdot \beta} \K C,\Delta}
    \end{prooftree}
    \end{center}
    $\byIH$
    \begin{center}
    \begin{prooftree}[small]
        \hypo{}
        \ellipsis{$d_1'$}{G \sepC X \update_\alpha \Gamma, B_1 \sep \update_\alpha \Imp A \sepC Y \update_{\alpha\cdot \beta} \K C,\Delta}
        \hypo{\cdots}
        \hypo{}
        \ellipsis{$d_n'$}{G \sepC X \update_\alpha \Gamma \update_{\alpha \cdot \beta} C \Imp \sep \update_\alpha \Imp A \sepC Y \update_{\alpha\cdot \beta} \K C,\Delta}
        \infer3[(L$\K_3$)]{G \sepC X \update_\alpha \Gamma \sep \update_\alpha \Imp A \sepC Y \update_{\alpha\cdot \beta} \K C,\Delta}
    \end{prooftree}
    \end{center}
    
    \bigskip
    
    \noindent$-$ Announcement rules: we need to consider both versions of each rule. We first consider (R$[\cdot])_1$, whilst (L$[\cdot])_1$ is treated likewise. If $G \sepC X \update_\alpha \, A, M \Imp N \update_{\alpha \cdot A} \, M' \Imp N',B$ is an axiom, then so is $G\sepC X \update_\alpha\,  M \Imp N, [A]B \update_{\alpha \cdot A} M' \Imp N'$. Otherwise, we consider the last applied rule $ \Rule$. If $ \Rule$ has $[A]B$ as its principal formula, then it must be (R$[\cdot])_1$ and the premise is precisely $G\sepC X \update_\alpha \, A,M \Imp N \update_{\alpha \cdot A} M' \Imp N',B$. If $[A]B$ is not principal, we apply the inductive hypothesis on the premise(s) and then $ \Rule$ again. We show one case with (L$\K _3$), whilst others are treated analogously. In the following, $\overline{Y}$ denotes $Y \update_{\alpha \cdot A \cdot   \beta} \K C,\Delta$.

    \begin{center}
    \begin{prooftree}[compact]
            \hypo{}
            \ellipsis{$d_1$}{ G \sepC X \update_\alpha \Gamma,  [A]B\update_{\alpha \cdot A} \Gamma', B_1 \sep \overline{Y}}
            \hypo{\cdots}
            \hypo{}
            \ellipsis{$d_k$}{ G \sepC X \update_\alpha  \Gamma,  [A]B  \update_{\alpha \cdot A} \Gamma' \update_{\update \cdot A \cdot \beta} C \Imp \sep \overline{Y} }
        \infer3[L($\K_3$)]{G \sepC X \update_\alpha \Gamma, [A]B \update_{\alpha \cdot A} \Gamma' \sep Y \update_{\alpha \cdot A \cdot   \beta} \K C,\Delta}
    \end{prooftree}
    \end{center}
    $\byIH$
    \begin{center}
    \begin{prooftree}[compact]
        \hypo{}
        \ellipsis{$d_1'$}{  G \sepC X \update_\alpha  A, \Gamma \update_{\alpha \cdot A} \Gamma', B,B_1 \sep \overline{Y}}
        \hypo{\cdots}
        \hypo{}
        \ellipsis{$d_k'$}{ G \sepC X \update_\alpha  A, \Gamma \update_{\alpha \cdot A} \Gamma',B \update_{\alpha \cdot A \cdot \beta} C \Imp \sep \overline{Y} }
        \infer3[L($\K_3$)]{G \sepC X \update_\alpha  A, \Gamma \update_{\alpha \cdot A} \Gamma',B \sep Y \update_{\alpha \cdot A \cdot \beta} \K C, \Delta}
    \end{prooftree}
    \end{center}

    We now turn to (R$[\cdot])_2$, while (L$[\cdot])_2$ is treated likewise. The proof proceeds similarly, except when the dynamic sequent $\update_{\alpha\cdot A} \ \Gamma'$ appears in the premise(s) of $\Rule$. In such cases, we use hp-invertibility of (R$[\cdot])_1$ on the premise(s) displaying $\update_{\alpha\cdot A} \ \Gamma'$ instead of the inductive hypothesis. Afterwards, we might need to apply a different version of rule $\Rule$, \emph{i.e.} (L$[\cdot])_1$ (resp. (R$[\cdot])_1$) instead of (L$[\cdot])_2$ (resp. (R$[\cdot])_2$), and (L$\K_2$) instead of (L$\K_3$). We provide some examples, whilst other cases can be treated similarly.
    
    If $\Rule$ is an occurrence of (L$[\cdot])_2$, we proceed as follows.
    \begin{center}
    \begin{prooftree}
        \hypo{}
        \ellipsis{$d_1$}{G\sep X \update_\alpha \Gamma,[A]B,A}
        \hypo{}
        \ellipsis{$d_2$}{G \sep X \update_\alpha \Gamma,[A]B \update_{\alpha \cdot A} C \Imp}
        \infer2[(L$[\cdot])_2$]{G \sep X \update_\alpha [A]C,\Gamma,[A]B}
    \end{prooftree}
    \end{center}
    
    \noindent We apply the inductive hypothesis on the left-hand premise, to obtain $G\sep X \update_\alpha \ A,\Gamma,A \update_{\alpha\cdot A} \ \Imp B$ while we use hp-invertibility of (R$[\cdot])_1$ on the right-hand premise to get $G \sep X \update_\alpha \ A,\Gamma \update_{\alpha \cdot A} \ C \Imp B$. We then apply (L$[\cdot])_1$ to obtain the desired conclusion:
    
    \begin{center}
    \begin{prooftree}
        \hypo{}
        \ellipsis{$d_1'$}{G\sep X \update_\alpha A,\Gamma,A \update_{\alpha\cdot A} \Imp B}
        \hypo{}
        \ellipsis{$d_2'$}{G \sep X \update_\alpha A,\Gamma \update_{\alpha \cdot A} C \Imp B}
        \infer2[(L$[\cdot])_1$]{G \sep X \update_\alpha A,[A]C,\Gamma \update_{\alpha \cdot A} \Imp B}
    \end{prooftree}\\
    \end{center}

    When $\Rule$ is an occurrence of (L$\K_3$) with two premises as below

\begin{center}
    \begin{prooftree}[regular]
        \hypo{}
        \ellipsis{$d_1$}{G \sepC X \update_\alpha M \Imp N, [A]B, A \sepC Y \update_{\alpha \cdot A} \K C, \Delta}
        \hypo{}
        \ellipsis{$d_2$}{G \sepC X \update_\alpha  M \Imp N, [A]B  \update_{\alpha \cdot A} C \Imp \sepC Y \update_{\alpha \cdot A} \K C, \Delta }
        \infer2[L($\K_3$)]{G \sepC X \update_\alpha  M \Imp N, [A]B \sepC Y \update_{\alpha \cdot A} \K C, \Delta}
    \end{prooftree}
\end{center}

    \noindent we consider the right-hand premise only: we apply hp-invertibility of (R$[\cdot])_1$ on it and then use (L$\K_2$) instead of (L$\K_3$), to get\footnote{Note that this method generalises straightforwardly to an application of (L$\K_3$) with more premises.}:
    \begin{displaymath}
    \begin{prooftree}[regular]
        \hypo{}
        \ellipsis{$d_2'$}{ G \sepC X \update_\alpha A, M \Imp N \update_{\alpha \cdot A} C \Imp B  \sepC Y \update_{\alpha \cdot A} \K C, \Delta }
        \infer1[L($\K_2$)]{ G \sepC X \update_\alpha A, M \Imp N \update_{\alpha \cdot A} \Imp B \sepC Y \update_{A} \K C, \Delta}
    \end{prooftree}
    \end{displaymath}

    Eventually, when $\Rule$ is an occurrence of (New) where $\update_{\alpha\cdot A} \ \Gamma'$ appears in the premise, hp-invertibility of (R$[\cdot])_1$ directly provides the conclusion, with no need for further application of (New), as shown below.
    \begin{center}
    \begin{prooftree}
        \hypo{}
        \ellipsis{$d$}{G \sepC X \update_\alpha \Gamma, [A]B \  \update_{\alpha\cdot A} \Imp \update_{\alpha\cdot A\cdot C} \ \Gamma'}
        \infer1[(New)]{G \sepC X \update_\alpha \Gamma, [A]B \update_{\alpha\cdot A\cdot C}\ \Gamma'}
    \end{prooftree}
    $\overset{(R[\cdot])_1}{\byInvertibility}$\footnote{We write $\overset{\Rule}{\byInvertibility}$ in an informal proof to indicate that we use the hp-invertibility of rule $\Rule$.}
    \begin{prooftree}
        \hypo{}
        \ellipsis{$d'$}{G \sepC X \update_\alpha A,\Gamma \ \update_{\alpha\cdot A} \ \Imp B \update_{\alpha \cdot A \cdot C} \ \Gamma'}
    \end{prooftree}
    \end{center}

    \hfill \\
    \noindent$-$ Dynamic rules: (Lat) and (Rat), as well as (Recall) are hp-invertible by hp-admissibility of (IW). (New) is hp-invertible by hp-admissibility of (DW).

\end{proof}

From now on, for any given rule $\Rule$, we denote by $\Rule^\star$ the inverted rule\footnote{In case rule $\Rule$ has several premises then $\Rule^\star$ denotes any of the corresponding inverted rules, e.g. for (R$\land$), we denote by (R$\land)^\star$ either $\begin{prooftree}[small]
    \hypo{G\sep X \update_\alpha M \Imp N,A\land B}
    \infer1{G\sep X\update_\alpha M \Imp N, A}
\end{prooftree}$ or
$\begin{prooftree}[small]
    \hypo{G\sep X \update_\alpha M \Imp N,A\land B}
    \infer1{G\sep X\update_\alpha M \Imp N, B}
\end{prooftree}$.}, that is derivable, by Lemma \ref{lemmaInvertible}. \\


As is standard in a hypersequent framework, there is a merge rule that can be shown to be hp-admissible.
%
To define the rule, we need to introduce the symbol $\xMerge$ for $X\xMerge Y$, which intuitively combines two dynamic hypersequents $X$ and $Y$ in the following way: all dynamic sequents labelled by the same sequence of announcements $\alpha$ are merged pairwise, while the remaining dynamic sequents are simply juxtaposed.


\begin{definition} Let $X$ and $Y$ be two DSs, then  $X \xMerge Y$ is inductively defined as follows:
\begin{itemize}
    \item if $X$ is empty (resp. $Y$ is empty) then $X \xMerge Y := Y$ (resp. $X \xMerge Y := X$);
    \item if $X= X' \update_{\alpha}\,  M \Imp N$ and $Y= Y' \update_{\alpha} \, P \Imp Q$ then $X \xMerge Y := X'\xMerge Y' \update_{\alpha} M,P\Imp N,Q$;
    \item if $X= X' \update_{\alpha} M \Imp N$ and $Y= Y' \update_{\beta} P \Imp Q$, where $\alpha\notin \ labels(Y')$ and $\beta \notin \ labels(X')$, then $X \xMerge Y := X'\xMerge Y' \update_{\alpha} M\Imp N \update_{\beta}P \Imp Q$.
\end{itemize}
\end{definition}

\noindent For instance, if $X = \Gamma \update_{A} \Gamma' \update_{A\cdot B} \Gamma''$ and $Y= \Delta \update_A \Delta' \update_B \Delta'' \update_{B\cdot A} \Delta'''$, then $X \xMerge Y = \Gamma\Delta \update_A \Gamma'\Delta' \update_{A\cdot  B} \Gamma''  \update_B \Delta'' \update_{B\cdot A} \Delta'''$.

\begin{lemma}
    The rule of \emph{Merge} is hp-admissible.
    \begin{align*}
    \begin{prooftree}[regular]
        \hypo{G \sep X \update_\alpha M \Imp N \sep Y \update_\alpha P \Imp Q}
        \infer1[\emph{(Merge)}]{G \sep X \xMerge Y \update_\alpha M,P \Imp N,Q}
    \end{prooftree}
    \end{align*}
\end{lemma}

\begin{proof}
    The proof proceeds by induction on the height of the derivation of the premise. If $G \sepC X \update_\alpha M \Imp N \sep Y \update_\alpha P \Imp Q$ is an axiom, so is $G \sepC X \xMerge Y \update_\alpha M,P \Imp N,Q$. Otherwise, we consider the last applied rule $ \Rule$: we apply the inductive hypothesis to the premise(s) of $ \Rule$ and then use $ \Rule$ again. We only show two paradigmatic examples.

    If $\Rule$ is an occurrence of (R$\K$):
    \begin{align*}
        &\begin{prooftree}[regular]
            \hypo{}
            \ellipsis{$d$}{G \sep X \update_\alpha\Gamma \sep Y \update_\alpha \Delta \sep \update_\alpha \Imp A}
            \infer1[(R$\K$)]{G \sep X \update_\alpha \Gamma, \K A \sep Y \update_\alpha \Delta}
        \end{prooftree}
        \quad 
        \byIH
        \quad
        \begin{prooftree}[regular]
            \hypo{}
            \ellipsis{$d'$}{G \sep X \xMerge Y \update_\alpha \Gamma\Delta  \sep \update_\alpha \Imp A}
            \infer1[(R$\K$)]{G \sep X \xMerge Y \update_\alpha \Gamma\Delta, \K A}
        \end{prooftree}        
    \end{align*}

\medskip If $ \Rule$ is an occurrence of (L$\K_3$) of the following shape,
    \begin{center}
    \begin{prooftree}[inter]
        \hypo{}
        \ellipsis{$d_1$}{G \sepC X \update_\alpha \Gamma, B_1 \sep Y \update_\alpha \Delta \update_{\alpha \cdot \beta} \K A, \Delta'}
        \hypo{\cdots}
        \hypo{}
        \ellipsis{$d_n$}{G \sepC X \update_\alpha \Gamma \update_{\alpha\cdot \beta} A \Imp \sep Y \update_\alpha \Delta \update_{\alpha \cdot \beta} \K A, \Delta'}
        \infer3[(L$\K_3$)]{  G \sepC X \update_\alpha \Gamma \sep Y \update_\alpha \Delta \update_{\alpha \cdot \beta} \K A, \Delta'}
    \end{prooftree}
    \end{center}    

    \noindent we simply apply the inductive hypothesis to the rightmost premise, and conclude with (L$\K_1$):
    \begin{center}
    \begin{prooftree}[regular]
        \hypo{}
        \ellipsis{$d_n'$}{G \sep X \xMerge Y \update_\alpha \Gamma \Delta \update_{\alpha\cdot \beta} A, \K A, \Delta'}
        \infer1[(L$\K_1$)]{  G \sep X \xMerge Y \update_\alpha \Gamma \Delta \update_{\alpha\cdot \beta} \K A, \Delta'}
    \end{prooftree}
    \end{center}   

\end{proof}


\begin{lemma}\label{lemmaAdmissibilityC}
The rules of \emph{contraction} are hp-admissible.
\begin{align*}
&\begin{prooftree}[regular]
    \hypo{G \sep X \update_\alpha A, A, M \Imp N}
    \infer1[(LC)]{G \sep X \update_\alpha A,M \Imp N}
\end{prooftree}
&
&\begin{prooftree}[regular]
    \hypo{G \sep X \update_\alpha M \Imp N,A,A}
    \infer1[(RC)]{G \sep X \update_\alpha M \Imp N, A}
\end{prooftree}
\end{align*}
\end{lemma}

\begin{proof}
    We simultaneously prove the hp-admissibility of (LC) and (RC) by induction on the height of the derivation of the premise of each rule. If $G \sepC X \update_\alpha A, A, M\Imp N$ and $G \sepC X \update_\alpha M \Imp N, A, A$ are axioms, then so are $G \sepC X \update_\alpha A, M \Imp N$ and $G \sepC X \update_\alpha M \Imp N, A$. Else, we distinguish cases on the last applied rule $ \Rule$. If neither of the occurrences of $A$ is principal, we apply the inductive hypothesis to the premise(s) of $ \Rule$, and then use $ \Rule$ again. If one of the occurrence of $A$ is principal, we start from the premise(s) of $ \Rule$ and use the hp-invertibility of the rules before applying the inductive hypothesis; then we use $ \Rule$ again. We only show one example, while the others can be treated similarly.

    If the last applied rule is (R$[\cdot])_2$, we first use hp-invertibility of (R$[\cdot])_1$, then we apply the inductive hypothesis twice, and eventually we use (R$[\cdot])_2$ again:
        \begin{align*}
        \begin{prooftree}[regular]
            \hypo{}
            \ellipsis{$d_1$}{ G \sep X \update_\alpha A, \Gamma [A]B \, \update_{\alpha \cdot A} \, \Imp B}
            \infer1[(R$[\cdot]$)]{  G \sep X \update_\alpha \Gamma, [A]B,[A]B}
        \end{prooftree}
        \overset{(R[\cdot])_1}{\dashrightarrow} \quad
        & G \sep X \update_\alpha A, A, \Gamma  \, \update_{\alpha \cdot A} \Imp B, B \\
        \byIH \quad
        & G \sep X \update_\alpha A, \Gamma   \, \update_{\alpha \cdot A} \Imp B, B \\
        \byIH \quad
        &\begin{prooftree}[regular]
            \hypo{}
            \ellipsis{$d_1'$}{ G \sep X \update_\alpha A, \, \Gamma \update_{\alpha \cdot A}\, \Imp B}
            \infer1[(R$[\cdot]$)]{ G \sep X \update_\alpha \, \Gamma,  [A]B}
        \end{prooftree} 
        \end{align*}
    
\end{proof}


We end this section by proving the admissibility of the rule $New^{A}$, which will be crucial to prove the admissibility of the Cut-rule. 

\begin{lemma}\label{lemmaAdmissibilityNewA}
The rule \emph{New}$^{A}$ is admissible.
\begin{displaymath}
\begin{prooftree}[regular]
\hypo{G \sep X \update_\alpha A, M \Imp N \update_{\alpha \cdot A} \Imp}
\infer1[(New$^A$)]{G \sep X \update_\alpha A, M \Imp N}
\end{prooftree}
\end{displaymath}
\end{lemma}

\begin{proof}
    We proceed by induction on the height of the derivation of the premise. If $G\sepC X \update_\alpha A, M \Imp N \update_{\alpha \cdot A} \Imp$ is an axiom then so is $G \sepC X \update_\alpha A, M \Imp N$. If $G\sepC X \update_\alpha A, M \Imp N \update_{\alpha \cdot A} \Imp$ is of the form $G \sepC X \update_\alpha A, M \Imp N \update_{\alpha \cdot A} \Imp \update_{\alpha \cdot A \cdot B} \, M' \Imp N'$---namely the dynamic sequent $X$ contains a sequent of the form $\update_{\alpha \cdot A \cdot B} \, M' \Imp N'$---we simply apply (New) to obtain the desired conclusion:

    \begin{center}
        \begin{prooftree}[regular]
            \hypo{G \sepC X \update_\alpha A, M \Imp N \update_{\alpha \cdot A} \Imp \update_{\alpha \cdot A \cdot B} \, M' \Imp N'}
            \infer1[(New)]{G \sepC X \update_\alpha A, M \Imp N \update_{\alpha \cdot A \cdot B} \, M' \Imp N'}
        \end{prooftree}
    \end{center}

    Finally, if $G\sepC X \update_\alpha A, M \Imp N \update_{\alpha \cdot A} \Imp$ is neither an axiom nor of the form indicated above, we consider the last applied rule $ \Rule$. If $ \Rule$ is an application of a right-rule, we apply the inductive hypothesis to the premise(s) of $ \Rule$ and then use $ \Rule$ again, as in the example below, where $\Rule$ is an occurrence of (R$\K$).
    \begin{displaymath}
        \begin{prooftree}[regular]
            \hypo{}
            \ellipsis{$d$}{G \sep X \update_\alpha A,M \Imp N \update_{\alpha \cdot A} \Imp \sep \update_\alpha \Imp B}
            \infer1[(R$\K$)]{G \sep X \update_\alpha A,M \Imp N, \K B \ \update_{\alpha \cdot A} \Imp }
        \end{prooftree}
    \
    \byIH
    \
        \begin{prooftree}[regular]
            \hypo{}
            \ellipsis{$d'$}{G \sep X \update_\alpha A,M \Imp N \sep \update_\alpha \Imp B}
            \infer1[(R$\K$)]{G \sep X \update_\alpha A,M \Imp N, \K B}
        \end{prooftree}
    \end{displaymath}

    If $ \Rule$ is an application of a left-rule that does not involve neither $A$ nor the dynamic sequent $\update_{\alpha \cdot A} \Imp$, we reason as before. Here is an example with (L$[\cdot]$):
    \begin{align*}
        &\begin{prooftree}[regular]
            \hypo{}
            \ellipsis{$d_1$}{G \sep X \update_\alpha A,\Gamma,B  \ \update_{\alpha \cdot A} \Imp}
            \hypo{}
            \ellipsis{$d_2$}{G \sep X \update_\alpha A,\Gamma \update_{\alpha \cdot A} \Imp \update_{\alpha \cdot B} \ C \Imp}
            \infer2[(L$[\cdot]$)]{G \sep X \update_\alpha [B]C,A,\Gamma \update_{\alpha \cdot A} \Imp }
        \end{prooftree}\\
    \\
    \byIH
    &\qquad \qquad
        \begin{prooftree}[regular]
            \hypo{}
            \ellipsis{$d_1'$}{G \sepC X \update_\alpha A,\Gamma,B}
            \hypo{}
            \ellipsis{$d_2'$}{G \sepC X \update_\alpha A,\Gamma \update_{\alpha \cdot B} C \Imp}
            \infer2[(L$[\cdot]$)]{G \sepC X \update_\alpha [B]C,A,\Gamma}
        \end{prooftree}
    \end{align*}

If $\Rule$ is an application of a left-rule that does involve either A or the dynamic sequent $\update_{\alpha \cdot A} \Imp$, then we obtain the desired result using hp-admissibility of weakening and contraction, applying the inductive hypothesis or directly using $\Rule$ again. We only show the most significant cases.\\
    
    \noindent $-$ $ \Rule$ is an application of (L$\K_3$) involving $\update_{\alpha\cdot A} \Imp$, with active dynamic sequent $\update_{\alpha\cdot A \cdot \beta} \K C, \Delta$ where $\beta = B_1 \cdots B_n$, as below

    \begin{center}
    \begin{prooftree}[regular]
        \hypo{}
        \ellipsis{$d_1$}{ G \sepC X \update_\alpha A, \Gamma \update_{\alpha \cdot A} \Imp B_1 \sep \overline{Y}}
        \hypo{\cdots}
        \hypo{}
        \ellipsis{$d_{n+1}$}{ G \sepC X \update_\alpha A, \Gamma \update_{\alpha \cdot A} \Imp \update_{\alpha \cdot A \cdot \beta} C \Imp \sep \overline{Y}}
        \infer3[(L$\K_3$)]{ G \sep X \update_\alpha A, \Gamma \update_{\alpha\cdot A} \Imp \sep Y \update_{\alpha \cdot A \cdot \beta} \K C, \Delta} 
    \end{prooftree}\\
    \end{center}
    
    \noindent we apply the inductive hypothesis on all the premises except the leftmost, and use $G \sepC X \update_\alpha A,  \Gamma,  A \sep Y \update_{\alpha \cdot A \cdot \beta} \K C, \Delta$ (that is derivable by Lemma \ref{lemmaDerivabilityTrivialSequents}) as a first premise to apply (L$\K_3$) with $n+2$ premises:
    
    \begin{center}    
        \begin{prooftree}[inter]
        \hypo{G \sepC X \update_\alpha A,  \Gamma,  A \sepC \overline{Y}}
        \hypo{}
        \ellipsis{$d_1$}{G \sepC X \update_\alpha A, \Gamma \update_{\alpha \cdot A} \Imp B_1 \sepC \overline{Y}}
        \hypo{\cdots}
        \hypo{}
        \ellipsis{$d_{n+1}'$}{G \sepC X \update_\alpha A, \Gamma \update_{\alpha \cdot A \cdot \beta} C \Imp \sepC \overline{Y}}
        \infer4[(L$\K_3$)]{G \sep X \update_\alpha A, \Gamma \sep Y \update_{\alpha \cdot A \cdot \beta} \K C, \Delta} 
        \end{prooftree}\\
    \end{center}
\hfill \\

    \noindent$-$ $ \Rule$ is an application of (L$[\cdot])$ where $A$ is the principal formula:
    \begin{center}
    \begin{prooftree}[regular]
            \hypo{}
            \ellipsis{$d_1$}{  G \sep X \update_{\alpha} M \Imp N, A \ \update_{[A]B} \Imp}
            \hypo{}
            \ellipsis{$d_2$}{  G \sep X \update_\alpha M \Imp N \ \update_{[A]B} \Imp \update_{\alpha \cdot A} \ B \Imp}
        \infer2[(L$[\cdot]$)]{G \sep X \update_\alpha [A]B, M \Imp N \ \update_{[A]B} \Imp}
    \end{prooftree}\\
    \end{center}
    
    \noindent We first use the hp-admissibility of (IW) to get
        \begin{align*}
            &\vdots \ d_1' & &\vdots \ d_2' \\
            \quad   G \sepC X \update_{\alpha} [A]B,M &\Imp N, A \ \update_{[A]B} \Imp \quad
            &G \sepC X \update_\alpha [A]B, M \Imp &N \ \update_{[A]B} \Imp \update_{\alpha \cdot A} \ B \Imp
        \end{align*}
    
    \noindent then we apply the inductive hypothesis, use (L$[\cdot])$ and then hp-admissibility of (LC) to conclude:

    \begin{center}
    \begin{prooftree}[regular]
            \hypo{}
            \ellipsis{$d_1''$}{G \sep X \update_{\alpha} [A]B, M \Imp N, A}
            \hypo{}
            \ellipsis{$d_2''$}{G \sep X \update_\alpha [A]B,M \Imp N \update_{\alpha \cdot A} B \Imp}
        \infer2[(L$[\cdot]$)]{G \sep X \update_\alpha [A]B, [A]B, M \Imp N}
        \infer1[(LC)]{G \sep X \update_\alpha [A]B, M \Imp N}
    \end{prooftree}
    \end{center}
    
    \end{proof}

\section{Soundness and Completeness}\label{sectionSoundAndComplete}

In this section, we show that $\calculus$ is sound and complete.


\begin{theorem}[Soundness]
     For all formulas $A \in \lPAL$, if $\vdash_{\calculus} \Imp A$ then $\satisfies_{PAL} A$.
\end{theorem}

\begin{proof}
    We need to show that axioms are valid and that all rules are correct, \emph{i.e.} that they preserve validity. As for the axioms, this is straightforward. As for the rules, for the sake of clarity, we only present in detail the case of the rule  (R$[\cdot])_1$. The other rules can be treated similarly.\\

    \noindent For (R$[\cdot])_1$, let $\alpha = A_1 \cdots A_n$ and suppose $\satisfies (G \sepC X \update_\alpha A, M \Imp N \update_{\alpha \cdot A} M' \Imp N',B)^\tau$ but $\nvDash (G \sepC X \update_\alpha M \Imp N, [A]B \update_{\alpha \cdot A} M' \Imp N')^\tau$. Then, there are a model $\ModelM$ and a world $w\in W$ such that
    \begin{align*}
        \M,w \nvDash G^\tau \lor \K (X^\tau \lor [\alpha] (\bigwedge M \imp \bigvee N \lor [A]B) \lor [\alpha][A](\bigwedge M' \imp \bigvee N') 
    \end{align*}
    so there is $u\in W$ such that $w\sim u$ and:
    \begin{align*}
        (1) \quad &\M,u \nvDash X^\tau \\
        (2) \quad &\M,u \nvDash [\alpha](\bigwedge M \imp \bigvee N \lor [A]B)\\
        (3)\quad &\M,u \nvDash [\alpha][A](\bigwedge M' \imp \bigvee N')
    \end{align*}

    From $(3)$, we get $\M,u \satisfies A_1, \M_{A_1},u \satisfies A_2, \cdots, \M_\alpha, u \satisfies A$ and $\M_{\alpha \cdot A} \nvDash \bigwedge M' \imp \bigvee N'$. From $(2)$ we get $\M_\alpha,u \satisfies \bigwedge M$ and  $\M_{\alpha},u \nvDash [A]B$. From that and $\M_\alpha, u \satisfies A$, we conclude that $\M_{\alpha \cdot A},u \nvDash B$.

    However, since $\satisfies (G \sepC X \update_\alpha A,M \Imp N \update_{\alpha \cdot A} M' \Imp N', B)^\tau$ and $\M,w \nvDash G^\tau$, necessarily
    \begin{align*}
        \M,w \satisfies \K (X^\tau \lor [\alpha](A \land \bigwedge M \imp \bigvee N) \lor [\alpha][A](\bigwedge M' \imp \bigvee N' \lor B)).
    \end{align*}
    Now, because $\M,u \nvDash X^\tau$, either $\M,u \satisfies [\alpha](A \land \bigwedge M \imp \bigvee N)$ or $M,u \satisfies [\alpha][A](\bigwedge M' \imp \bigvee N' \lor B)$.
    In the first case, we get a contradiction because $\M_\alpha,u \satisfies \bigwedge M$ and $\M_\alpha,u \satisfies A$ but $\M_\alpha,u \nvDash \bigvee N$. In the other case, we get a contradiction from $\M_{\alpha \cdot A},u \nvDash \bigwedge M' \imp \bigvee N'$ and $\M_{\alpha\cdot A},u \nvDash B$ which implies $\M_{\alpha \cdot A} \nvDash \bigwedge M' \imp \bigvee N' \lor B$, hence $M,u \nvDash [\alpha][A](\bigwedge M' \Imp \bigvee N' \lor B)$.\\

    Therefore $\satisfies (G \sepC X \update_\alpha M \Imp N,[A]B \update_{\alpha \cdot A} M' \Imp N')^\tau$.

\end{proof}


We now move to the completeness of the calculus. In order to prove that $\calculus$ is complete with respect to PAL, we need to  show that all axioms of PAL are derivable, as well as that the inference rules of PAL are admissible. This can be standardly carried out except for the case of the axiom of announcement composition, which requires us to show a  derivability lemma (see Lemma \ref{lemmaComposition}). On the other hand,  in order to prove this derivability lemma, we first need an enhanced  notion of complexity of formulas, \emph{i.e.} a notion of complexity of formulas that is relative to the sequents formulas belong to in the context of a dynamic hypersequent.\footnote{This complexity measure will also be used later to prove Cut-admissibility in Section \ref{sectionCut}.} We thus also introduce the notion of \emph{DHS-complexity} in Definition \ref{defComplexity}. Finally,   note that many other calculi for PAL (\emph{e.g.} \cite{nomura2015revising},\cite{balbiani2010tab}) require some preliminary lemmas in order to show a completeness result; hence our procedure is standard in the context of proof theory for public announcement logic.

    \begin{definition}[DHS-Complexity]\label{defComplexity}
    The DHS-complexity of a formula $A$ \emph{relatively to a sequence of announcements} $\alpha$, $c(A,\alpha)$, is defined as:
    \begin{align*}
        c(A,\alpha) := c(A) + c(\alpha),
    \end{align*}
    where $c(A)$ is inductively defined as follows:
    \begin{align*}
        c(p)&:= 1  \quad \text{ for } p\in P  \\
        c(\lnot A) &:= c(A) + 1  & c(\K A) &:= c(A) + 1\\
        c(A \land B) &:= c(A) + c(B) +1 & c([A]B) &:= c(A) + c(B) + 1
    \end{align*}
    whilst $c(\alpha)$ is inductively defined as follows:
    \begin{align*}
        c(\epsilon) &:= 0 &  c(\alpha \cdot A) &:= c(\alpha) + c(A).
    \end{align*}
    \end{definition}

    \begin{lemma}\label{lemmaComposition}
    Let $A, B, C$ be formulas of $\lPAL$, and $\alpha,\beta$ be finite sequences of announcements. Then, the following dynamic hypersequents are derivable:
    \begin{align*}
        (1) \quad &G \sep X \update_{\alpha\cdot A \cdot B \cdot \beta}\ C, M \Imp N \update_{\alpha \cdot A\land[A]B \cdot \beta} \ M' \Imp N',C \\
        (2) \quad &G \sep X \update_{\alpha\cdot A \cdot B \cdot \beta} \ M \Imp N, C \update_{\alpha \cdot A\land[A]B \cdot \beta} \ C, M' \Imp N'
    \end{align*}
\end{lemma}

\begin{proof}
    We simultaneously prove that $(1)$ and $(2)$ are derivable by induction on the sum of the DHS-complexity of each occurrence of C, denoted $\overline{c}(C):= c(C, \alpha\cdot A \cdot B \cdot \beta) + c(C, \alpha \cdot A \land [A]B \cdot \beta)$.

    Without loss of generality, we prove the result for empty $G$, $X$ and $\alpha$. In the following, $\beta = B_1 \cdot B_2 \cdots B_{n-1}$ and we write $\phi$ for $A \land [A]B$. We distinguish the following cases, where we only present the details for $(1)$, whilst $(2)$ is treated similarly.\\

    \noindent$\bullet$ Case $C= p$. To derive $(1)$, we (bottom-up) apply $2\times n$ times (New) to go from $\update_{\phi\cdot \beta} \ \Gamma$ to $\update_\epsilon \ \Gamma'$ and from $\update_{A\cdot B \cdot \beta} \ \Delta$ to $\update_{A} \ \Delta'$, then we apply $n+1$ times (Rat) and $n$ times (Lat) to obtain an axiom with $p\Imp p$. We write $\Rule^\ast$ when we apply several times rule $\Rule$.\\

    \begin{prooftree}[regular]
        \hypo{p \Imp p \update_A \ \Imp p \update_{A\cdot B} \ \Imp p \update \cdots \update_{A\cdot B \cdot \beta} M \Imp N, p \update_\phi \ p\Imp  \update \cdots\update_{\phi \cdot \beta} \ p, M' \Imp N'}
        \infer1[(Lat)$^\ast$ ($n$ times)]{\Imp p \update_A \ \Imp p \update_{A\cdot B} \ \Imp p \update \cdots \update_{A\cdot B \cdot \beta} M \Imp N, p  \update_\phi \ \Imp \update \cdots \update_{\phi \cdot \beta} \ p, M' \Imp N'}
        \infer1[(Rat)$^\ast$ ($n+1$ times)]{\Imp \update_A\  \Imp \update_{A\cdot B} \ \Imp \update \cdots \update_{A\cdot B \cdot \beta} M \Imp N, p  \update_\phi \ \Imp \update \cdots \update_{\phi \cdot \beta} \ p, M' \Imp N'}
        \infer1[(New)$^\ast$ ($2n$ times)]{\update_{A\cdot B \cdot \beta} M \Imp N, p \update_{\phi\cdot \beta}  \ p, M' \Imp N'}
    \end{prooftree}

    \hfill \\

    \noindent$\bullet$ Case $C =\lnot C_1$. Since $\overline{c}( C_1) < \overline{c}(\lnot C_1)$, we apply the inductive hypothesis. Then, we simply apply (L$\lnot$) and (R$\lnot$). This is straightforward.\\

    \noindent$\bullet$ Case $C = C_1 \land C_2$. We can apply the inductive hypothesis because $\overline{c}(C_i) < \overline{c}(C_1 \land C_2)$ for $i \in \lbrace 1, 2 \rbrace$. Then, we apply admissibility of (IW) and use (L$\land$) and (R$\land$). This is straightforward. \\

    \noindent$\bullet$ Case $C = \K C_1$. Note that $\overline{c}(C_1) < \overline{c}(\K C_1)$ and, for all $\beta'\cdot B_i \subseteq \beta$, $c(B_i, A \cdot B \cdot \beta') + c(B_i, A \land [A]B \cdot \beta') < \overline{c}(\K C_1)$. This allows us to use the inductive hypothesis to obtain the following DHSs:
    \begin{align*}
        (1) \quad &\update_{A \cdot B} \Imp B_1 \update_\phi B_1 \Imp \\
        (2) \quad &\update_{A \cdot B \cdot B_1} \Imp B_2 \update_{\phi \cdot B_1} B_2 \Imp \\
        \vdots \ \quad & \\
        (n-1) \quad &\update_{A \cdot B \cdots B_{n-2}} \Imp B_{n-1} \update_{\phi \cdots B_{n-2}} B_{n-1} \Imp \\
        (n) \quad &\update_{A \cdot B \cdot \beta} C_1 \Imp \update_{\phi \cdot\beta} \Imp C_1
    \end{align*}

    \noindent By admissibility of (IW), (EW) and (DW), from each $(i)$ we get the following $(i')$:
    \begin{align*}
        (1)' \ &\update_{A \cdot B \cdot \beta} \K C_1, \Gamma \update_{\phi \cdot \beta} \Gamma' \sep A \Imp \update_A B \Imp \update_{A \cdot B} \Imp B_1 \update_\phi B_1 \Imp \update \cdots \update_{\phi \cdot \beta} \Imp C_1 \\
        (2)' \ &\update_{A \cdot B \cdot \beta} \K C_1, \Gamma \update_{\phi \cdot \beta} \Gamma' \sep A \Imp \update_A B \Imp \update_{A \cdot B \cdot B_1} \Imp B_2 \update_{\phi} B_1 \Imp \update_{\phi \cdot B_1} B_2 \Imp \update \cdots \update_{\phi \cdot \beta} \Imp C_1 \\
        \vdots \quad & \\
        (n)' \ &\update_{A \cdot B \cdot \beta} \K C_1, \Gamma \update_{\phi \cdot \beta} \Gamma' \sep A \Imp \update_A B \Imp \update_{A\cdot B \cdot \beta} C_1 \Imp \update_\phi B_1 \Imp \update \cdots \update_{\phi \cdot \beta} \Imp C_1 
    \end{align*}

    \noindent To apply (L$\K_3$) with $\update_{A\cdot B\cdot \beta} \K C_1, \Gamma'$, we need $n+1$ premises, since the sequence of announcements here is $B \cdot B_1 \cdots B_{n-1}$. Hence, we also need the following DHS, that is derivable by Lemma \ref{lemmaDerivabilityTrivialSequents}:
    \begin{displaymath}
        (0)' \quad \update_{A \cdot B \cdot \beta} \K C_1, \Gamma \update_{\phi \cdot \beta} \Gamma' \sep A \Imp \update_A B \Imp B \update_\phi B_1 \Imp \update \cdots \update_{\phi \cdot \beta} \Imp C_1
    \end{displaymath}
    We now apply (L$\K_3$)  to $(0)'-(n)'$ and continue the derivation as follows, where $X := \update_{A\cdot B\cdot \beta} \ \K C_1, \Gamma \update_{\phi \cdot \beta}\ \Gamma'$:\\

    \begin{center}
    \begin{prooftree}[regular]
            \hypo{X \sep A \Imp A \update \cdots}
                \hypo{(0)'}
                \hypo{(1)'}
                \hypo{\cdots}
                \hypo{(n)'}
            \infer4[(L$\K_3$)]{ X \sep A\Imp \update_A B \Imp \update_{\phi} B_1 \Imp \update \cdots \update_{\phi \cdot \beta} \Imp C_1}
        \infer2[(L$[\cdot]$)]{\update_{A\cdot B\cdot \beta} \ \K C_1, \Gamma \update_{\phi \cdot \beta}\ \Gamma' \sep A,[A]B \Imp \update_{\phi} \ B_1 \Imp \update \cdots \update_{\phi \cdot \beta} \Imp C_1}
        \infer1[(L$\land$)]{\update_{A\cdot B\cdot \beta} \ \K C_1, \Gamma \update_{\phi \cdot \beta}\ \Gamma' \sep A \land [A]B \Imp \update_{\phi} \ B_1 \Imp \update \cdots \update_{\phi \cdot \beta} \Imp C_1}
        \infer1[(Recall)$^\ast$  ($n$ times)]{\update_{A\cdot B\cdot \beta} \ \K C_1, \Gamma \update_{\phi \cdot \beta}\ \Gamma' \sep  \Imp \update_{\phi} \Imp \update \cdots \update_{\phi \cdot \beta} \Imp C_1}
        \infer1[(New)$^\ast$  ($n$ times)]{\update_{A\cdot B\cdot \beta} \ \K C_1, \Gamma \update_{\phi \cdot \beta}\ \Gamma' \sep \update_{\phi \cdot \beta} \Imp C_1 }
        \infer1[(R$\K$)]{\update_{A\cdot B\cdot \beta} \ \K C_1, \Gamma \update_{\phi \cdot \beta}\ \Gamma', \K C_1}
    \end{prooftree}
    \end{center}

\hfill \\

\noindent$\bullet$ Case $C=[C_1]C_2$. We use the inductive hypothesis and admissibility of (DW). Note that the inductive hypothesis applies because $\overline{c}(C_1) < \overline{c}([C_1]C_2)$ and $c(C_2, A \cdot B \cdot \beta \cdot C_1) + c(C_2, A \land [A] B \cdot \beta \cdot C_1) < \overline{c}([C_1]C_2)$.\\

\noindent
\begin{prooftree}[compact]
        \hypo{\update_{A \cdot B \cdot \beta}  \, \Gamma,  C_1 \update_{\phi \cdot \beta} C_1, \Gamma'}
        \infer1[\footnotesize{(DW)}]{\update_{A \cdot B \cdot \beta} \, \Gamma,  C_1 \update_{\phi \cdot \beta} C_1, \Gamma' \update_{\phi \cdot \beta \cdot C_1} \Imp C_2}
        \hypo{\update_{A \cdot B \cdot \beta \cdot C_1}\,  C_2 \Imp \update_{\phi \cdot \beta \cdot C_1} \Imp C_2}
        \infer1[\footnotesize{(DW)}]{\update_{A \cdot B \cdot \beta} \, \Gamma \update_{\phi \cdot \beta} C_1, \Gamma' \update_{A \cdot B \cdot \beta \cdot C_1} C_2 \Imp \update_{\phi \cdot \beta \cdot C_1} \Imp C_2}
    \infer2[\footnotesize{(L$[\cdot]$)}]{\update_{A\cdot B\cdot \beta}\,  [C_1]C_2, \Gamma \ \update_{\phi \cdot \beta} C_1, \Gamma' \ \update_{\phi \cdot \beta \cdot C_1} \Imp C_2 }
    \infer1[\footnotesize{(R$[\cdot]$)}]{\update_{A\cdot B\cdot \beta} \, [C_1]C_2, \Gamma \ \update_{\phi \cdot \beta} \,  \Gamma', [C_1]C_2 }
\end{prooftree}
\end{proof}

Note that Lemma \ref{lemmaComposition} can be understood as a way of identifying dynamic sequents labelled by $A\cdot B$ and those labelled by $A \land [A]B$ (and more generally by $\alpha \cdot A \cdot B \cdot \beta$ and by $\alpha \cdot A \land [A]B \cdot \beta$). This result has a semantic counterpart. Indeed, it can be proved that, for all models $\ModelM$ and all formulas $A,B$, $W_{A\cdot B}= W_{A \land [A]B}$ (see for instance \cite{DEL}).

\begin{theorem}[Completeness]
     For all formulas $A \in \lPAL$, if $\satisfies_{PAL} A$ then $\vdash_{\calculus} \Imp A$.
\end{theorem}

    \begin{proof}
    Since \textbf{PAL} is complete w.r.t. the standard semantics of PAL, it is sufficient to demonstrate that all axioms of \textbf{PAL} are derivable in $\calculus$ and that the rules of modus ponens and necessitation are admissible. Modus ponens can be shown to be admissible using the cut-rule (see Section \ref{sectionCut}); whilst the necessitation rule can be shown to be admissible using external weakening. Derivations for the axioms of propositional logic and modal logic \textbf{S5} are obtained by direct (bottom-up) proof-search, see  \cite{Poggiolesi2008} for details. Derivations for the reduction axioms are also obtained by (bottom-up) proof search (except for the announcement composition axiom that we will consider separately). As an illustrative example, consider the axiom of announcement and knowledge, which we prove to be a theorem by means of the following derivation:

    \vspace{0.3cm}
    
    \begin{center}
    \begin{prooftree}[regular]
            \hypo{A \Imp A \sep A \Imp \update_A \Imp B}
            \hypo{A \Imp  \update_A \K B \Imp \sep A \Imp \update_A B \Imp B}
            \infer1[(L$\K_2$)]{A \Imp \update_A \K B \Imp \sep A \update_A \Imp B}
        \infer2[(L$[\cdot]$)]{A, [A]\K B  \Imp \sep A \Imp \update_A \Imp B}
        \infer1[(R$[\cdot]$)]{A, [A]\K B  \Imp \sep \Imp [A]B}
        \infer1[(R$\K$)]{A, [A]\K B \Imp \K [A]B}
        \infer1[(R$\imp$)]{[A]\K B \Imp A \imp \K [A]B}
        \infer1[(R$\imp$)]{\Imp [A]\K B \imp( A \imp \K [A]B)}
    \end{prooftree}
    \end{center}

    and

\begin{center}
    \begin{prooftree}[regular]
            \hypo{A \Imp A \update_A \Imp \K B}
                \hypo{A, \K [A]B \Imp  \update_A  \Imp \sep A \Imp A \update_A\Imp B}
                \hypo{A, \K [A]B \Imp  \update_A  \Imp \sep A \Imp  \update_A B \Imp B}
            \infer2[(L$[\cdot]$)]{A, \K [A] B\Imp \update_A \Imp \sep A, [A]B \Imp \update_A \Imp B}
            \infer1[(L$\K_2$)]{A, \K [A]B \Imp \update_A \Imp \sep A \Imp \update_A \Imp B}
            \infer1[(Recall)]{A, \K[A]B \Imp \update_A \Imp \sep \Imp \update_A \Imp B}
            \infer1[(New)]{A, \K[A]B \Imp \update_A \Imp \sep \update_A \Imp B}
            \infer1[(R$\K$)]{A, \K [A] B \Imp  \update_A \Imp \K B}
        \infer[separation = -5em]2[(L$\imp$)]{A, A \imp \K [A]B \Imp \update_A \Imp \K B}
        \infer1[(R$[\cdot]$)]{A \imp \K [A]B \Imp [A]\K B}
        \infer1[(R$\imp$)]{\Imp (A \imp \K [A]B) \imp [A]\K B}
    \end{prooftree}
\end{center}

    \hfill \\

    We now move to the axiom of announcements composition $[A][B]C \eq [A\land [A]B]C$. Providing a direct proof for that axiom schema requires Lemma \ref{lemmaComposition} that demonstrates derivability of dynamic hypersequents $\update_{A \cdot B}  \Imp C \update_{A \land [A]B} C \Imp$ and $\update_{A \cdot B}  C \Imp \update_{A \land [A]B} \Imp C$, for all formulas $C$. We only show the derivation for one direction, whilst the other can be obtained in a similar way.\\

    \noindent
    \begin{prooftree}[small]
                \hypo{A \Imp A \update_A B \Imp \update_{A\cdot B} \Imp C}
                \hypo{A, A \Imp \update_A B \Imp B \update_{A\cdot B} \Imp C}
                \infer1[(R$[\cdot]$)]{A \Imp [A]B \update_A B \Imp \update_{A\cdot B} \Imp C}
            \infer2[(R$\land$)]{A \Imp A \land [A]B \update_A B \Imp \update_{A \cdot B} \Imp C}
            \hypo{}
            \ellipsis{Lemma \ref{lemmaComposition}}{A \Imp \update_A \ B \Imp \update_{A \cdot B} \Imp C \update_{A \land [A]B} C \Imp}
        \infer2[(L$[\cdot]$)]{A, [A\land [A]B]C \Imp \update_{A} \ B \Imp \update_{A \cdot B} \Imp C}
        \infer1[(R$[\cdot]$)]{A, [A\land [A]B]C \Imp \update_A \Imp [B]C}
        \infer1[(R$[\cdot]$)]{[A\land [A]B]C \Imp [A][B]C}
        \infer1[(R$\imp$)]{\Imp [A\land [A]B]C \imp [A][B]C}
    \end{prooftree}\\

\end{proof}


\section{Cut Admissibility}\label{sectionCut}

In this section we demonstrate that the cut-rule is admissible in the calculus $\calculus$\footnote{Hereafter, for clarity and succinctness, we may omit explicit notations of derivations of premises $\vdots \ d$.}.  We first prove the following lemma.

\begin{lemma}\label{lemmaPermutationRules}
    Propositional rules, modal rules, announcement rules, as well as (New), (Recall) and (Lat) permute down with respect to $n$ applications of (Rat). Furthermore, any occurrence of (Rat) also permutes down with respect to $n$ applications of (Rat) when the atom $p$ that is removed in its conclusion is not active in either of these $n$ applications of (Rat).
\end{lemma}

\begin{proof}
Let us first consider the permutation with 1-premise rules, that is straightforward. We show the permutability of $n$ consecutive applications of (Rat) with the rule (R$[\cdot])$, while other cases can be treated similarly. In the following, (Rat)$^\ast$ denotes $n$ applications of (Rat).

    \begin{center}
        \begin{prooftree}[regular]
            \hypo{G \sep M \Imp N, p \update \cdots \update_\alpha A, M' \Imp N', p \ \update_{\alpha \cdot A} \Imp B}
            \infer1[(R$[\cdot]$)]{G \sep M \Imp N, p \update \cdots \update_\alpha  M' \Imp N', p, [A]B }
            \ellipsis{(Rat)$^\ast$}{G \sep M \Imp N \update \cdots \update_\alpha  M \Imp N', p, [A]B}
        \end{prooftree}
        
        \bigskip
        $\downarrow$
        \bigskip
        
        \begin{prooftree}[regular]
            \hypo{G \sep M \Imp N, p \update \cdots \update_\alpha A, M' \Imp N', p \ \update_{\alpha \cdot A} \Imp B}
            \ellipsis{(Rat)$^\ast$}{G \sep M \Imp N \update \cdots \update_\alpha  A, M' \Imp N', p \ \update_{\alpha \cdot A} \Imp B}
            \infer1[(R$[\cdot]$)]{G \sep M \Imp N \update \cdots \update_\alpha  M' \Imp N', p, [A]B } 
        \end{prooftree}
    \end{center}

    Let us now consider the permutation with n-premises rules (for $n \geq 2$). As a way of paradigmatic example, we show that rule (L$\K_3$) permutes down with $n$ applications of (Rat). Without loss of generality, let us assume the first premise has been obtained by $n$ applications of (Rat). We then proceed as follows (where $\beta = B_1 \cdots B_{k-1}$):

    \begin{center}
    \begin{prooftree}[regular]
        \hypo{G \sep \Gamma,  p \update \cdots \update_\alpha \Gamma', p, B_1 \sep Y \update_{\alpha \cdot \beta} \K C, \Delta}
        \ellipsis{(Rat)$^\ast$}{G \sep \Gamma \update \cdots \update_\alpha \Gamma', p, B_1 \sep Y \update_{\alpha \cdot \beta} \K C, \Delta}
        \hypo{\cdots}
        \hypo{}
        \ellipsis{$d$}{G \sep \Gamma \update \cdots \update_\alpha \Gamma', p \update_{\alpha \cdot \beta} C \Imp \sep Y \update_{\alpha \cdot \beta} \K C, \Delta}
        \infer3[(L$\K_3$)]{G \sep \Gamma \update \cdots \update_\alpha \Gamma', p \sep Y \update_{\alpha \cdot \beta} \K C, \Delta}
    \end{prooftree}
    \end{center}
    
    \noindent Here, we use hp-admissibility of weakening on all the premises but the leftmost one to obtain the required premises and apply (L$\K_3$) before applying (Rat) $n$ times.
    \begin{center}
    \begin{prooftree}[regular]
        \hypo{G \sep \Gamma,  p \update \cdots \update_\alpha \Gamma', p, B_1 \sep Y \update_{\alpha \cdot \beta} \K C, \Delta}
        \hypo{\cdots}
        \hypo{}
        \ellipsis{$d'$}{G \sep \Gamma,  p \update \cdots \update_\alpha \Gamma', p \update_{\alpha \cdot \beta} C \Imp \sep Y \update_{\alpha \cdot \beta} \K C, \Delta}
        \infer3[(L$\K_3$)]{G \sep \Gamma,  p \update \cdots \update_\alpha \Gamma', p \sep Y \update_{\alpha \cdot \beta} \K C, \Delta}
        \ellipsis{(Rat)$^\ast$}{G \sep \Gamma \update \cdots \update_\alpha \Gamma', p \sep Y \update_{\alpha \cdot \beta} \K C, \Delta}
    \end{prooftree}
    \end{center}
\end{proof}

\begin{lemma}\label{lemmaPermutationRules2}
    Propositional rules, modal rules, announcement rules, as well as (New), (Recall) and (Rat) permute down with respect to $n$ applications of (Lat).
    
    Furthermore, any occurrence of (Lat) also permutes down with respect to $n$ applications of (Lat) when the atom $p$ that is removed in its conclusion is not active in either of these $n$ applications of (Lat).
\end{lemma}

\begin{proof}
    The proof is similar to that for (Rat).
\end{proof}

\begin{theorem}
\label{thm:cut}
The rule of \emph{Cut} is admissible.

    \begin{displaymath}
    \begin{prooftree}[regular]
        \hypo{}
        \ellipsis{$d_1$}{G \sep X \update_\alpha M \Imp N, A}
        \hypo{}
        \ellipsis{$d_2$}{H \sep Y \update_\alpha A, P \Imp Q}
        \infer2[(Cut)]{G \sep H \sep X \xMerge Y \update_\alpha M,P \Imp N,Q}
    \end{prooftree}
    \end{displaymath}

\end{theorem}

\begin{proof}
    The proof proceeds by main induction on the DHS-complexity of the cut-formula (see Definition \ref{defComplexity}), and secondary induction on the sum of the heights of the derivations $d_1$ and $d_2$ of each premise. We distinguish cases on the last rule applied on the left premise.\footnote{Note that in this proof, we denote by $\Rule^\ast$ any finitely many applications of rule $\Rule$.}

\medskip\noindent $-$ If $G \sepC X \update_\alpha M \Imp N, A$ is an axiom, then either the conclusion is an axiom as well, or the cut can be replaced by several applications of the rules (IW), (EW) and (DW) on $H \sepC Y \update_\alpha p, P \Imp Q$ so to obtain the conclusion.

\medskip\noindent $-$ If $G \sepC X \update_\alpha M \Imp N, A$ has been obtained by a rule $\mathcal{R}$ where the formula $A$ is not principal,  then we apply the inductive hypothesis to the premise(s) of $\mathcal{R}$ and then $ \Rule$ again.\footnote{Note that, in some cases, we might need to apply another version of the rule, \emph{e.g.} (L$\K_2$) instead of (L$\K_3$) or (R$[\cdot])_1$ instead of (R$[\cdot])_2$, or no rule at all if $\Rule$ is (New).}

\medskip\noindent $-$ If $G \sepC X \update_\alpha M \Imp N, A$ has been obtained by a rule $\mathcal{R}$ where the formula $A$ is principal,  then we distinguish the following subcases. (i) The rule $\mathcal{R}$ is a propositional rule, (ii) the rule $\mathcal{R}$ is a modal rule, (iii) the rule $\mathcal{R}$ is an announcement rule, (iv) the rule $\mathcal{R}$ is a dynamic rule.
   
\medskip(i) As an example, we consider the case where $\mathcal{R}$ is the rule (R$\neg$), we have:
\begin{center}
    \begin{prooftree}[regular]
            \hypo{ G \sepC X \update_\alpha A, M \Imp N}
            \infer1[(R$\lnot$)]{ G \sepC X \update_\alpha  M \Imp N, \lnot A}
            \hypo{  H \sepC Y \update_\alpha \lnot A, P \Imp Q}
        \infer2[(Cut)]{G \sepC H \sepC X \xMerge Y \update_\alpha M, P \Imp N, Q}
    \end{prooftree}
\end{center}

\noindent Here, we consider the right premise of the Cut. We first apply hp-invertibility of (L$\lnot$) to obtain $H \sepC Y \update_\alpha P \Imp Q,A$ and then cut $A$, where this application of (Cut) is admissible by induction on the DHS-complexity of $A$:

\begin{center}
    \begin{prooftree}[regular]
            \hypo{ G \sepC X \update_\alpha A, M \Imp N}
            \hypo{  H \sepC Y \update_\alpha P \Imp Q, A}
        \infer2[(Cut)]{G \sepC H \sepC X \xMerge Y \update_\alpha M, P \Imp N, Q}
    \end{prooftree}
\end{center}

\medskip(ii) Consider the case where $\mathcal{R}$ is the rule (R$\Box)$. We then need to look at the right-hand premise $H \sepC Y \update_\alpha \K A, P \Imp Q$. If it is an axiom, then the conclusion is an axiom as well. If it has been obtained by a rule $\Rule'$ where $\K A$ is not the principal formula, we use the inductive hypothesis---with induction on the sum of the heights---and then apply $\Rule'$ again. If finally $H \sepC Y \update_\alpha \K A, P \Imp Q$ has been obtained by a rule $\Rule'$ which is either (L$\Box_{1})$ or (L$\Box_{2})$ and has $\Box A$ as principal formula, we apply the inductive hypothesis twice, to cut $\K A$ and then $A$, with first induction of the sum of the heights and second on the DHS-complexity of the cut-formula. We only analyse the case of the rule (L$\K_2$), whilst that of (L$\K_1$) is dealt with analogously. We have the following situation:

\begin{center}
\begin{prooftree}[regular]
	\hypo{G \sep X \update_\alpha \Gamma \sep \update_\alpha \Imp A}
	\infer1[(R$\K$)]{G \sep X \update_\alpha \Gamma,\K A}
 	\hypo{H \sep Y \update_\alpha \K A, \Delta \sep Z \update_\alpha A, \Lambda}
 	\infer1[(L$\K_2$)]{H \sep Y \update_\alpha \K A, \Delta \sep Z \update_\alpha \Lambda}
 	\infer2[(Cut)]{G \sep H \sep X \xMerge Y \update_\alpha \Gamma \Delta \sep Z \update_\alpha \Lambda}
\end{prooftree}
\end{center}

\noindent We first apply (Cut) on the right-hand premise and on $G \sep X \update_\alpha \Gamma,\K A$. This application of (Cut) is admissible by inductive hypothesis on the sum of the heights:

\begin{center}
\begin{prooftree}[regular]
	\hypo{G \sep X \update_\alpha \Gamma, \K A}
	\hypo{H \sep Y \update_\alpha \K A, \Delta \sep Z \update_\alpha A, \Lambda}
	\infer2[(Cut)]{G \sep H \sep X \xMerge Y \update_\alpha \Gamma \Delta \sep Z \update_\alpha A, \Lambda}	
\end{prooftree}
\end{center}

\noindent We then cut the remaining $A$, with an application of (Cut) that is admissible by induction on the DHS-complexity of the cut-formula. Eventually, the desired conclusion is obtained by several applications of the contraction rules as well as (Merge):

\begin{center}
\begin{prooftree}[regular]
	\hypo{G \sep X \update_\alpha \Gamma \sep \update_\alpha \Imp A}
	\hypo{G \sep H \sep X \xMerge Y \update_\alpha \Gamma \Delta \sep Z \update_\alpha A, \Lambda}
	\infer2[(Cut)]{G \sep G \sep H \sep X \xMerge Y \update_\alpha \Gamma \Delta \sep X \update_\alpha \Gamma \sep Z \update_\alpha \Lambda}	
	\infer1[(Merge)$^\ast$+(C)$^\ast$]{G \sep H \sep X \xMerge Y \update_\alpha \Gamma \Delta \sep Z \update_\alpha \Lambda}
\end{prooftree}
\end{center}

\medskip Suppose now that the rule $\mathcal{R}^{\prime}$ is (L$\K_3$) and the cut-formula is principal in (L$\K_3$). For the sake of simplicity, we show how to proceed with an example where only three premises of the rule (L$\K_3$) occur. The generalisation to $n$ premises can be found in the Appendix.\footnote{Here, $\overline{H} := H \sep Y \update_{B \cdot C} \K A, \Delta$.}

    \begin{center}
    \begin{prooftree}[regular]
            \hypo{ G \sep X \update_{B\cdot C} \Gamma \sepC \update_{B\cdot C} \Imp A}
            \infer1[(R$\K$)]{ G \sep X \update_{B\cdot C}  \Gamma,  \K A}
                \hypo{}
                \ellipsis{$(a)$}{  \overline{H}  \sep \Lambda,  B}
                \hypo{}
                \ellipsis{$(b)$}{\overline{H}  \sep  \Lambda \update_B \Imp C}
                \hypo{}
                \ellipsis{$(c)$}{\overline{H}  \sep  \Lambda \update_{B \cdot C} A \Imp}
            \infer3[(L$\K_3$)]{ H \sep Y \update_{B\cdot C}  \K A, \Delta \sep \Lambda}
            \infer2[(Cut)]{G \sepC H \sep X \xMerge Y \update_{B \cdot C} \Gamma \Delta \sepC \Lambda}
        \end{prooftree}
    \end{center}
   
    We first work with premise $(c)$ in the following way:
    \begin{center}
        \begin{prooftree}[regular]
            \hypo{G \sepC X \update_{B\cdot C} \Gamma \sepC \update_{B\cdot C} \Imp A}
                    \hypo{ G \sepC X \update_{B\cdot C}  \Gamma,  \K A}
                    \hypo{}
                    \ellipsis{$(c)$}{  H \sepC Y \update_{B \cdot C} \K A, \Delta  \sepC \Lambda \update_{B \cdot C} A \Imp}
                \infer2[(Cut)]{G \sepC H \sepC X\xMerge Y \update_{B\cdot C} \, \Gamma\Delta \sep \Lambda \update_{B\cdot C} A \Imp}
            \infer2[(Cut)]{G \sepC G \sepC H \sepC X \update_{B\cdot C} \Gamma \sep X \xMerge Y \update_{B \cdot C} \Gamma \Delta \sep \Lambda \ \update_{B \cdot C} \Imp}
            \infer1[(Merge)$^\ast+$(C)$^\ast$]{G \sepC H \sep X \xMerge Y \update_{B \cdot C} \Gamma \Delta \sep \Lambda \ \update_{B \cdot C} \Imp}
            \infer1[(DW)]{G \sepC H \sep X \xMerge Y \update_{B \cdot C} \Gamma \Delta \sep \Lambda \update_B C \Imp \update_{B \cdot C} \Imp}
            \infer1[(New$^A$)]{G \sepC H \sep X \xMerge Y \update_{B \cdot C} \Gamma \Delta \sep \Lambda \ \update_{B} C \Imp}
        \end{prooftree}
    \end{center}
    
    \medskip\noindent where the first cut is admissible by induction on the sum of the heights, whilst the second by induction on DHS-complexity of the cut-formula. Then we re-use the previously obtained DHS in order to apply the same method to premise (b):

    \begin{center}
        \begin{prooftree}[regular]
            \hypo{G \sepC H \sep X \xMerge Y \update_{B \cdot C} \Gamma \Delta \sep \Lambda \ \update_{B} C \Imp}
                    \hypo{ G \sepC X \update_{B\cdot C}  \Gamma,  \K A}
                    \hypo{}
                    \ellipsis{$(b)$}{  H \sepC Y \update_{B \cdot C} \K A, \Delta  \sepC \Lambda \ \update_{B} \ \Imp C}
                \infer2[(Cut)]{G \sepC H \sepC X\xMerge Y \update_{B\cdot C} \Gamma\Delta \sep \Lambda \ \update_{B}  \ \Imp C}
            \infer2[(Cut)]{G \sepC G \sepC H \sepC H \sepC X\xMerge Y \update_{B\cdot C} \Gamma\Delta \sep X \xMerge Y \update_{B \cdot C} \Gamma \Delta \sep \Lambda \Lambda \ \update_B \ \Imp}
            \infer1[(Merge)$^\ast+$(C)$^\ast$]{G \sepC H \sep X \xMerge Y \update_{B \cdot C} \Gamma \Delta \sep \Lambda \ \update_B \ \Imp}
            \infer1[(IW)]{G \sepC H \sep X \xMerge Y \update_{B \cdot C} \Gamma \Delta \sep B,  \Lambda\  \Imp \update_B \ \Imp}
            \infer1[(New$^A$)]{G \sepC H \sep X \xMerge Y \update_{B \cdot C} \Gamma \Delta \sep B,  \Lambda}
        \end{prooftree}\\
    \end{center}
    
    \medskip\noindent where the first cut is admissible by induction on the sum of the heights, whilst the second by induction of DHS-complexity of the cut-formula. Finally, we reach the desired conclusion with the following derivation:

        \begin{prooftree}[regular]
            \hypo{G \sepC H \sepC X \xMerge Y \update_{B \cdot C} \Gamma \Delta \sep B,  \Lambda}
                    \hypo{G \sepC X \update_{B\cdot C}  \Gamma,  \K A}
                    \hypo{}
                    \ellipsis{$(a)$}{H \sepC Y \update_{B \cdot C} \K A, \Delta  \sepC  \Lambda,  B}
                \infer2[(Cut)]{G \sepC H \sepC X\xMerge Y \update_{B\cdot C} \Gamma\Delta \sep  \Lambda,  B}
            \infer2[(Cut)]{G \sepC G \sepC H \sepC H \sepC X\xMerge Y \update_{B\cdot C} \Gamma\Delta \sep X \xMerge Y \update_{B \cdot C} \Gamma \Delta \sep \Lambda \Lambda}
            \infer1[(Merge)$^\ast+$(C)$^\ast$]{G \sepC H \sep X \xMerge Y \update_{B \cdot C} \Gamma \Delta \sep \Lambda}
        \end{prooftree}

    \hfill\\

    \medskip\noindent where the first cut is admissible by induction on the sum of the heights, whilst the second by induction of DHS-complexity of the cut-formula.

    \medskip(iii) Consider the case where $\mathcal{R}$ is the rule (R$[\cdot])_2$.\footnote{The case where $\mathcal{R}$ is (R$[\cdot])_1$ can be treated in an analogous, though simpler manner.} We then need to look at the right-hand premise $H \sepC Y \update_\alpha [A]B, P \Imp Q$. If it is an axiom, or it has been obtained by a rule $\mathcal{R}^{\prime}$ where $[A]B$ is not the principal formula, then the procedure is the standard one. Consider then the case where $\mathcal{R}^{\prime}$ is the rule (L$[\cdot])$. We have the following situation:\footnote{We consider the case where $\mathcal{R}^{\prime}$ is  (L$[\cdot])_2$; the case where  $\mathcal{R}^{\prime}$ is (L$[\cdot])_1$ can be treated analogously.}

    \begin{center}
    \begin{prooftree}[regular]
        \hypo{ G \sepC X \update_\alpha A, \Gamma \update_{\alpha \cdot A} \Imp B}
        \infer1[(R$[\cdot])_2$]{  G \sepC X \update_\alpha  \Gamma,  [A]B}
            \hypo{  H \sepC Y \update_\alpha  \Delta,  A}
            \hypo{  H \sepC Y \update_\alpha \Delta \update_{\alpha \cdot A} B \Imp}
            \infer2[(L$[\cdot])_2$]{  H \sepC Y \update_\alpha [A]B, \Delta}
        \infer2[(Cut)]{G \sepC  H \sepC X \xMerge Y \update_\alpha \Gamma \Delta}
    \end{prooftree}\\
    \end{center}

\medskip\noindent We proceed in the following way:

    \begin{center}
    
    \begin{prooftree}[regular]
        \hypo{ H \sepC Y \update_\alpha  \Delta,  A}
            \hypo{  G \sepC X \update_\alpha A, \Gamma \update_{\alpha \cdot A} \Imp B}
            \hypo{  H \sepC Y \update_\alpha \Delta \update_{\alpha \cdot A} B \Imp}
            \infer2[(Cut)]{G \sepC H \sepC X \xMerge Y \update_\alpha A, \Gamma \Delta \update_{\alpha \cdot A} \Imp}
            \infer1[(New$^A$)]{G \sepC H \sepC X \xMerge Y \update_\alpha A, \Gamma \Delta}
        \infer2[(Cut)]{G \sepC H \sepC H \sepC X \xMerge Y \xMerge Y \update_\alpha \Gamma \Delta \Delta}
        \infer1[(Merge)$^\ast +$ (C)$^\ast$]{G \sepC H \sepC X \xMerge Y \update_\alpha \Gamma \Delta}
    \end{prooftree}
    \end{center}
    \hfill \\
    
    \noindent where the first cut is admissible by induction on the sum of the heights, whilst the second is admissible by induction on the DHS-complexity of the cut-formula.
    
\medskip(iv) Consider the case where $\mathcal{R}$ is the rule (Rat)---the cases of the other dynamic rules can be treated straightforwardly. We have the following situation, where $\alpha = A_1 \cdots A_n$ and for all $1 \leq i \leq n$, $\alpha_i = A_1 \cdots A_i$:

    \begin{center}       
    \begin{prooftree}[regular]
        \hypo{}
        \ellipsis{$d_1'$}{G \sep X \update \cdots \update_{\alpha_{n-1}} \Gamma_{n-1}, p \update_{\alpha_n} \Gamma_n, p}
        \infer1[(Rat)]{G \sep X \update_{\alpha_{n-1}}  \Gamma_{n-1} \update_{\alpha_n} \Gamma_n, p}
        \hypo{}
        \ellipsis{$d_2$}{H \sep Y \update_{\alpha_n} p, \Delta_n}
        \infer2[(Cut)]{G \sep H \sep X \xMerge Y \update \cdots \update_{\alpha_n} \Gamma_n\Delta_n}
    \end{prooftree}\\
    \end{center}

    \noindent Here we go up derivation $d_1'$ until we find an application of a rule $\mathcal{R}^{\prime\prime}$ different from (Rat) on that atom $p$. There are two possibilities:

    \begin{itemize}
        \item[$(1)$] After $m$ applications of (Rat), we indeed find a rule $\mathcal{R}^{\prime\prime}$ different from (Rat) on that atom $p$. Thanks to Lemma \ref{lemmaPermutationRules}, we can permute $\mathcal{R}^{\prime\prime}$ down with the $m$ applications of (Rat). We thus obtain a DHS that we can use, together with $H \sepC Y \update_{\alpha_n} \ p, \Delta_n$ to cut atom $p$. This cut is admissible by induction on the sum of heights. We then apply rule $\mathcal{R}^{\prime\prime}$ again\footnote{ As already mentioned, in some cases, we might need to apply another version of the rule---\emph{e.g.} (L$\K_2$) instead of (L$\K_3$)---or no rule at all---in case of (New).}, to obtain the desired result.  
    \end{itemize}

    \begin{itemize}
        \item[$(2)$] Alternatively, there is no such rule because $G \sepC X \update \cdots \update_{\alpha_{n-1}} \Gamma_{n-1}, p \update_{\alpha} \Gamma_n, p$ is an axiom.

    \begin{center}       
    \begin{prooftree}[regular]
        \hypo{G \sep X' \update \cdots \update_{\alpha_{n-1}} \Gamma_{n-1}, p \update_{\alpha_n} \Gamma_n, p}
        \ellipsis{(Rat)$^\ast$}{G \sep X \update \cdots \update_{\alpha_{n-1}} \Gamma_{n-1}, p \update_{\alpha_n} \Gamma_n, p}
        \infer1[(Rat)]{G \sep X \update \cdots \update_{\alpha_{n-1}} \Gamma_{n-1} \update_{\alpha_n} \Gamma_n, p}
        \hypo{}
        \ellipsis{$d_2$}{H \sep Y \update_{\alpha_{n-1}} \Delta_{n-1} \update_{\alpha_n} p, \Delta_n}
        \infer2[(Cut)]{G \sep H \sep X \xMerge Y \update \cdots \update_{\alpha_n} \Gamma_n\Delta_n}
    \end{prooftree}\\
    \end{center}
        
        We distinguish three cases, that we analyse in the following:
        \begin{enumerate}
            \item[$(2a)$] $p$ is principal in $\Gamma_n$;
            \item[$(2b)$] $p$ is principal in another dynamic sequent $\Gamma_i$ (for $1 \leq i < n$);
            \item[$(2c)$] $p$ is principal in none of them.
        \end{enumerate}
    \end{itemize}

    $(2a)$ We have $G \sepC X' \update \cdots \update_{\alpha_{n-1}} \Gamma_{n-1}, p \update_{\alpha_n} p,\Gamma_n, p$ so we start from the right-hand premise $H \sepC Y \update_{\alpha_{n-1}} \Delta_{n-1} \update_{\alpha_n} p, \Delta_n$ and obtain the conclusion by admissibility of (IW), (EW) and (DW). \\

    $(2b)$ Let $1 \leq i < n$. The axiom we have is $G \sep X' \update \cdots \update_{\alpha_i} p,\Gamma_i, p \update \cdots \update_{\alpha_n} \Gamma_n, p$. Here, we need to look at the right-hand premise and distinguish the following cases:
    \begin{itemize}
        \item[$(2bi)$] If it is an axiom, where $p$ is principal in $\Delta_n$, we start from the left-hand premise $G \sep X \update \cdots \update_{\alpha_{n-1}} \Gamma_{n-1} \update_{\alpha_n} \Gamma_n, p$ and obtain the conclusion by admissibility of the weakening rules.
        \item[$(2bii)$] If it is an axiom with a principal formula different from $p$ in $\Delta_n$, then the conclusion is also an axiom.
        \item[$(2biii)$] If it has been obtained by a rule $\mathcal{R^{\prime}}$ that does not operate on $p$, then we consider the premise of $\mathcal{R}^{\prime}$, and we apply a cut between such premise and $G \sep X \update \cdots \update_{\alpha_n} \Gamma_n, p$. The cut is admissible by induction on the sum of heights. We then apply $\mathcal{R}^{\prime}$ again to obtain the desired conclusion.
        \item[$(2biv)$] If it is the result of $l$ applications of (Lat), then, as before, we go up the derivation until we find a rule other than (Lat) on the atom $p$. There are two possibilities that are analogous to (1) and (2) above. Let us call them (1$^{\prime}$) and (2$^{\prime}$), respectively. As for (1$^{\prime}$), we proceed analogously to (1). As for (2$^{\prime}$), we have the situation where after $l$ applications of (Lat), we find an axiom. We need to distinguish cases again.
            \begin{enumerate}
                \item[$(2^{\prime}a)$] $p$ is principal in $\Delta_n$, we proceed as in $(2a)$;
                \item[$(2^{\prime}b)$] $p$ is principal in a dynamic sequent labelled by the same sequence $\alpha_i$ as the axiom in the left-hand premise. Then the conclusion is an axiom.
                \item[$(2^{\prime}c)$] $p$ is principal in a dynamic sequent with a label $\alpha_j$ different from $\alpha_i$, \emph{i.e.} where $i\neq j$. Without loss of generality, let us consider $i < j$. In such a case, we start from the axioms, removing the $p$ in $\Gamma_{k}$ (resp. $\Delta_k$) for all $k > j$. We apply the inductive hypothesis on the sum of heights, and then apply (Lat) or (Rat) as many times as needed to obtain the conclusion. Below is an example.
            \end{enumerate}
        \end{itemize}

    \begin{center}
    \begin{prooftree}[regular]
                \hypo{p, \Gamma_1, p \update_{ A} \Gamma_2, p \update_{ A \cdot B} \Gamma_3, p \update_{A \cdot B \cdot C} \Gamma_4, p}
            \infer1[(Rat)]{p, \Gamma_1 \update_A \Gamma_2, p \update_{A \cdot B} \Gamma_3, p \update_{A \cdot B \cdot C} \Gamma_4, p}
            \infer1[(Rat)]{p, \Gamma_1 \update_A \Gamma_2 \update_{A \cdot B} \Gamma_3, p \update_{A \cdot B \cdot C} \Gamma_4, p}
            \infer1[(Rat)]{p, \Gamma_1 \update_A \Gamma_2 \update_{A \cdot B} \Gamma_3 \update_{A \cdot B \cdot C} \Gamma_4, p}
                \hypo{\Delta_1 \update_A p, \Delta_2, p\update_{A \cdot B} p, \Delta_3 \update_{A \cdot B \cdot C} p, \Delta_4}
            \infer1[(Lat)]{\Delta_1 \update_A \Delta_2, p \update_{A \cdot B} p, \Delta_3 \update_{A \cdot B \cdot C} p, \Delta_4}
            \infer1[(Lat)]{\Delta_1 \update_A \Delta_2, p \update_{A \cdot B} \Delta_3 \update_{A \cdot B \cdot C} p, \Delta_4}
        \infer2[(Cut)]{p,\Gamma_1 \Delta_1 \update_A \Gamma_2 \Delta_2, p \update_{A \cdot B} \Gamma_3 \Delta_3 \update_{A \cdot B \cdot C} \Gamma_4 \Delta_4}
    \end{prooftree}
    \end{center}

    \noindent We proceed as explained above, where $i=1$ and $j=2$: we start from the axioms wherein we remove $p$ from $\Gamma_3,\Gamma_4,\Delta_3$ and $\Delta_4$. We then cut $p$ in $\Gamma_2$ and $\Delta_2$ and apply (Rat) to remove the atom $p$ that still needs to be removed in $\Gamma_1$:

    \begin{center}
        \begin{prooftree}[regular]
        \hypo{p, \Gamma_1, p \update_A \Gamma_2, p \update_{A \cdot B} \Gamma_3 \update_{A \cdot B \cdot C} \Gamma_4 }
        \hypo{\Delta_1 \update_A p, \Delta_2, p \update_{A \cdot B} \Delta_3 \update_{A \cdot B \cdot C} \Delta_4}
        \infer2[(Cut)]{p, \Gamma_1 \Delta_1, p \update_A \Gamma_2 \Delta_2, p \update_{A \cdot B} \Gamma_3 \Delta_3 \update_{A \cdot B \cdot C} \Gamma_4 \Delta_4}
        \infer1[(Rat)]{p, \Gamma_1 \Delta_1 \update_A \Gamma_2 \Delta_2, p  \update_{A \cdot B} \Gamma_3 \Delta_3 \update_{A \cdot B \cdot C} \Gamma_4 \Delta_4}
    \end{prooftree}
    \end{center}

    The method generalises naturally: when (Rat) has been applied $m$ times, and (Lat) $l$ times, if $i < j$, then $m > l$. Therefore, after having applied the inductive hypothesis, we need to apply $(m-l)$ times (Rat) to remove the remaining atoms $p$ in sequents $\Gamma_k$ for all $i \leq k < j$.\\

    $(2c)$ Then the conclusion is an axiom so we can just start from it.
    
\end{proof}

\begin{corollary}
    $\calculus$ enjoys the subformula property: all formulas occurring in a $\calculus$-derivation ending with the dynamic hypersequent $G$ are subformulas of one of the formulas occurring in  $G$ or of one of the formulas occurring as announcement in $G$.
\end{corollary}

\begin{proof} Straightforward from the rules of the calculus $\calculus$ and Theorem \ref{thm:cut}.
\end{proof}

\section{Conclusion and future work}\label{sectionConclusion}

In this paper, we introduced a syntactic method for extending hypersequents to their dynamic counterparts. On the basis of this method, we developed a calculus for Public Announcement Logic, one of the most prominent and widely studied dynamic logics. The calculus was shown to satisfy several fundamental properties, including soundness, completeness, and cut-elimination. These results are of intrinsic interest, but they also suggest a number of directions for further research. First, we wish to establish additional metatheoretical results for the calculus of PAL, such as decidability or  interpolability. In both cases, it seems reasonable to believe that standard techniques -- such as the appeal to minimal derivations (see e.g. \cite{Poggiolesi2008, Poggiolesi2010}) in the first case, and adaptation of Maehara's method for Craig interpolation (see \cite{Kuznetsinterpolation_hypersequents,KUZNETSmulticomponent_method_interpolation}) in the second -- can be adapted quite naturally to the framework of dynamic hypersequents. We also wish to investigate issues related to complexity and countermodels extraction.
Second, the framework of dynamic hypersequents appears well suited for the systematic development of sequent calculi for a broader class of dynamic epistemic logics involving different dynamic modalities. As already noted, dynamic hypersequents significantly extend standard hypersequents by making it possible to represent all possible epistemic PAL updates of a model. Their structure is therefore genuinely rich. It is precisely this richer structure, which can be further extended to account for other kinds of epistemic updates, that we consider sufficiently robust to capture additional dynamic epistemic logics.

\section*{Acknowledgements.}
We would like to thank the referee for the valuable remarks, which greatly helped improve the paper.

\nocite{*}
\printbibliography

@inbook{avron,
	author = {Avron, Arnon},
	title = {The method of hypersequents in the proof theory of propositional non-classical logic},
	booktitle = {Logic: from Foundations to Applications},
	year = {1996},
	publisher = {Oxford University Press},
	pages = {1--32},
	editor = {W. Hodges and M. Hyland and C. Steinhorn and J. Strauss}
}

@inproceedings{balbiani2007tableau,
  author={Balbiani, Philippe and van Ditmarsch, Hans and Herzig, Andreas and De Lima, Tiago},
  title={A tableau method for public announcement logics},
  booktitle={Automated Reasoning with Analytic Tableaux and Related Methods: 16th International Conference, TABLEAUX 2007, Aix en Provence, France, July 3-6, 2007. Proceedings 16},
  pages={43--59},
  year={2007},
  organization={Springer}
}

@article{balbiani2010tab,
	title = {Tableaux for public announcement logics},
	author = {Balbiani, Philippe and van Ditmarsch, Hans and Herzig, Andreas and de Lima, Tiago},
	year = {2010},
	journal = {Journal of Logic and Computation},
	volume = {20},
	number = {1},
	publisher = {Oxford Univeristy Press},
	pages = {55--76}
}

@inproceedings{baltagetal:1998,
    author = {Alexandru Baltag and Lawrence S. Moss and S{\l}awomir Solecki},
    title = {The Logic of Public Announcements, Common Knowledge, and 
Private Suspicions},
    booktitle = {Proc.\ of 7th TARK},
    pages = {43--56},
    year = {1998},
    doi = {10.1007/978-3-319-20451-2_38},
}

@article{baltagandco,
  title={Epistemic Actions as Resources},
  author={Baltag, Alexandru and Coecke, Bob and Sadrzadeh, Mehrnoosh},
  journal={Journal of Logic and Computation},
  volume={17},
  number={3},
  pages={555--585},
  year={2007},
  }

@article{dyckhoffandco,
  title={Algebra, proof theory and applications for an intuitionistic logic of propositions, actions and adjoint modal operators},
  author={Dyckhoff, Roy and Sadrzadeh, Mehrnoosh and Truffaut, Julien},
  journal={ACM Transactions on Computational Logic (TOCL)},
  volume={14},
  number={34},
  pages={1--37},
  year={2013},
  }

@book{DEL,
  title={Dynamic epistemic logic},
  author={Ditmarsch, Hans van and Hoek, Wiebe and Kooi van Der, Barteld},
  volume={337},
  year={2007},
  publisher={Springer Science \& Business Media}
}

@article{frittellaandco,
  title={A proof-theoretic semantic analysis of dynamic epistemic logic},
  author={Frittella, Sabine and Greco, Giusepppe and  Kurz, Alexander and Palmigiano, Alessandra and Sikimić, Vlasta},
  journal={Journal of Logic and Computation},
  volume={26},
  number={6},
  pages={1961--2015},
  year={2016},
  }

@incollection{Kuznetsinterpolation_hypersequents,
	author = {Roman Kuznets},
	booktitle = {Concepts of Proof in Mathematics, Philosophy, and Computer Science},
	editor = {Dieter Probst and Peter Schuster},
	pages = {193--214},
	publisher = {De Gruyter},
	title = {Craig Interpolation Via Hypersequents},
	year = {2016},
    doi = {10.1515/9781501502620-012},
}

@article{KUZNETSmulticomponent_method_interpolation,
title = {Multicomponent proof-theoretic method for proving interpolation properties},
journal = {Annals of Pure and Applied Logic},
volume = {169},
number = {12},
pages = {1369-1418},
year = {2018},
note = {Logic Colloquium 2015},
issn = {0168-0072},
doi = {https://doi.org/10.1016/j.apal.2018.08.007},
author = {Roman Kuznets},
}

@inproceedings{liu2023,
  title={Non-labelled Sequent Calculi of Public Announcement Expansions of {K}45 and {S}5},
  author={Liu, Sizhuo and Sano, Katsuhiko},
  booktitle={International Workshop on Logic, Rationality and Interaction},
  pages={190--206},
  year={2023},
  organization={Springer}
}

@inproceedings{nomura2015revising,
  title={Revising a labelled sequent calculus for public announcement logic},
  author={Nomura, Shoshin and Sano, Katsuhiko and Tojo, Satoshi},
  booktitle={Structural Analysis of Non-Classical Logics: The Proceedings of the Second Taiwan Philosophical Logic Colloquium},
  pages={131--157},
  year={2015},
  organization={Springer}
}

@inproceedings{plaza1989,
  title={Logics of public announcements},
  author={Plaza, Jan},
  booktitle={Proceedings 4th international symposium on methodologies for intelligent systems},
  pages={201--216},
  year={1989}
}

@article{Poggiolesi2008,
  title={A Cut-Free Simple Sequent Calculus
for Modal Logic {S}5},
  author={Poggiolesi, Francesca},
  journal={Review of Symbolic Logic},
  year={2008},
pages={3-15},
volume={1}
}

@book{Poggiolesi2010,
    author={Poggiolesi, Francesca},
    title={Gentzen Calculi for Modal Propositional Logic},
    publisher={Springer},
    year={2010},
    place={Berlin},
    collection={Trends in Logic}
}

@article{Poggiolesi2013fromSingleToMany,
  title={From single agent to multi-agent via hypersequents},
  author={Poggiolesi, Francesca},
  journal={Logica Universalis},
  volume={7},
  number={2},
  pages={147--166},
  year={2013},
  publisher={Springer}
}

@article{KULE,
  title={Grafting hypersequents onto nested sequents},
  author={Kuznets, Roman and Lellmann, Bjorn},
  journal={Logic Journal of IGPL},
  volume={24},
  number={3},
  pages={375--423},
  year={2016}
  }

@inbook{Restall,
	author = {Restall, Greg},
	title = {Proofnets for S5: sequents and circuits for modal logic},
	booktitle = { Logic Colloquium 2005: Proceedings of the Annual European Summer Meeting of the Association for Symbolic Logic},
	year = {2007},
	publisher = {Cambridge University Press},
	pages = {151--172},
	editor = {C. Dimitracopoulos and L. Newelski and D. Normann}
}

@book{BasicProofTheoryTroelstra,
    author={Troelstra, Anne S. and Schwichtenberg, Helmut},
    title={Basic Proof Theory},
    publisher={Cambridge University Press},
    year={2000},
    place={Cambridge},
    edition={2},
    series={Cambridge Tracts in Theoretical Computer Science},
    collection={Cambridge Tracts in Theoretical Computer Science}
}

@article{wu2023labelled,
  title={A Labelled Sequent Calculus for Public Announcement Logic},
  author={Wu, Hao and van Ditmarsch, Hans and Chen, Jinsheng},
  journal={Studies in Logic},
  volume={16},
  number={3},
  pages={89--107},
  year={2023}
}

\appendix
\renewcommand\appendixname{Appendix}
\renewcommand\appendixpagename{Appendix}
\appendixpage

Let us explain the general method to prove admissibility of (Cut) when the left-hand premise has been obtained by an application of (R$\K$) and the right-hand premise has been obtained by an application of (L$\K_3$) with $k$ premises, for $\beta = B_1 \cdots B_{k-1}$. Let $(1), \cdots, (k)$ denote each premise of (L$\K_3$), let $(i)$ denote $G \sepC X \update_\alpha \Gamma \sepC \update_\alpha \Imp A$, and $(ii)$ $G \sepC X \update_\alpha  \Gamma,  \K A$. The derivation has this shape\footnote{Here $\overline{H} := H \sep Y \update_{\alpha \cdot \beta} \K A, \Delta$.}:

    \begin{center}
        \begin{prooftree}[regular]
            \hypo{}
            \ellipsis{}{G \sepC X \update_{\alpha\cdot \beta} \Gamma \sepC \update_\alpha \Imp A}
            \infer1[(R$\K$)]{G \sepC X \update_{\alpha\cdot \beta}  \Gamma,  \K A}
                \hypo{}
                \ellipsis{$(1)$}{\overline{H} \sepC Z \update_\alpha  \Lambda,  B_1}
                \hypo{\cdots}
                \hypo{}
                \ellipsis{$(k)$}{\overline{H} \sepC Z \update_\alpha \Lambda \update_{\alpha \cdot \beta} A \Imp}
            \infer3[(L$\K_3$)]{H \sepC Y \update_{\alpha \cdot \beta}  \K A, \Delta \sep Z \update_\alpha \Lambda}
            \infer2[(Cut)]{G \sepC H \sepC X \xMerge Y \update_{\alpha \cdot \beta} \Gamma \Delta \sepC Z \update_\alpha \Lambda}
        \end{prooftree}
    \end{center}
    \hfill \\
  
    \noindent We first apply $k$ times (Cut) between $(ii)$ and each $(j)$ premise (for $1 \leq j \leq k$), where each of these applications is admissible, by induction on the sum of heights. We denote by $(I_1), \cdots (I_k)$ the DHSs obtained this way:

    \begin{center}  
    \begin{prooftree}
        \hypo{}
        \ellipsis{$(ii)$}{ G \sepC X \update_{\alpha\cdot \beta} \Gamma, \K A}
        \hypo{}
        \ellipsis{$(1)$}{ H \sepC Y \update_{\alpha\cdot \beta} \K A, \Delta \sep Z \update_\alpha  \Lambda,  B_1}
        \infer2[(Cut)]{(I_1): \quad G \sepC H \sepC X \xMerge Y \update_{\alpha \cdot \beta} \Gamma \Delta \sep Z \update_\alpha  \Lambda,  B_1}
    \end{prooftree}\\

    \bigskip
    
    $\vdots$ \\
    
    \hfill\\
    
    \begin{prooftree}[regular]
        \hypo{}
        \ellipsis{$(ii)$}{ G \sepC X \update_{\alpha \cdot \beta} \Gamma, \K A}
        \hypo{}
        \ellipsis{$(k)$}{ H \sepC Y \update_{\alpha\cdot \beta} \K A, \Delta \sep Z \update_\alpha \Lambda \update_{\alpha \cdot \beta} A \Imp}
        \infer2[(Cut)]{(I_k): \quad G \sepC H \sepC X \xMerge Y \update_{\alpha \cdot \beta} \Gamma \Delta \sep Z \update_\alpha \Lambda \update_{\alpha \cdot \beta} A \Imp}
    \end{prooftree}\\
    \end{center}

    \noindent We now apply the inductive hypothesis between $(i)$ and $(I_k)$, with induction on the DHS-complexity of the cut-formula. We then use admissibility of (Merge), contraction rules, (DW) and (New$^A$) as follows, to obtain $(J_k): G \sep H \sep X \xMerge Y \update_{\alpha \cdot \beta} \Gamma \Delta \sep Z \update_\alpha \Lambda \update_{\alpha \cdot B_1 \cdots B_{k-2}} B_{k-1} \Imp $.

    \begin{center}
    \begin{prooftree}[regular]
            \hypo{}
            \ellipsis{$(i)$}{ G \sep X \update_{\alpha\cdot \beta} \Gamma \sep \update_{\alpha \cdot \beta} \Imp A}
            \hypo{}
            \ellipsis{$(I_k)$}{G \sepC H \sepC X \xMerge Y \update_{\alpha \cdot \beta} \Gamma \Delta \sep Z \update_\alpha \Lambda \update_{\alpha \cdot \beta} A \Imp}
        \infer2[(Cut)]{G \sepC G \sepC H \sepC X \sepC X \xMerge Y \update_{\alpha \cdot \beta}\  \Gamma \Delta \sep Z \update_\alpha \Lambda \update_{\alpha \cdot \beta} \Imp }
        \infer1[(Merge)$^\ast+$ (C)$^\ast$]{G \sepC H \sepC X \xMerge Y \update_{\alpha \cdot \beta} \ \Gamma \Delta \sep Z \update_\alpha \Lambda \update_{\alpha \cdot \beta} \Imp}
        \infer1[(DW)]{ G \sepC H \sepC X \xMerge Y \update_{\alpha \cdot \beta} \ \Gamma \Delta \sep Z \update_\alpha \Lambda \update_{\alpha \cdot B_1 \cdots B_{k-2}} B_{k-1} \Imp \update_{\alpha \cdot \beta} \Imp}
        \infer1[(New$^A$)]{(J_k): \quad G \sep H \sep X \xMerge Y \update_{\alpha \cdot \beta} \Gamma \Delta \sep Z \update_\alpha \Lambda \update_{\alpha \cdot B_1 \cdots B_{k-2}} B_{k-1} \Imp}
    \end{prooftree}
    \end{center}
    
    \hfill \\

    \noindent We now take $(I_{k-1})$ and $(J_k)$ and follow the same procedure: we apply the inductive hypothesis to cut $B_{k-1}$, and then use admissibility of the structural rules, thereby obtaining DHS $(J_{k-1}) : G \sepC H \sepC X \xMerge Y \update_{\alpha \cdot \beta} \Gamma \Delta \sep Z \update_\alpha \Lambda \update_{\alpha \cdot B_1 \cdots B_{k-3}} B_{k-2}\Imp$. The induction is here on the DHS-complexity of the cut-formula. The derivation proceeds as follows\footnote{Here $\beta':= B_1 \cdots B_{k-2}$.}:\\

    \begin{center}
    \begin{prooftree}[regular]
            \hypo{}
            \ellipsis{$(I_{k-1})$}{ G \sepC H \sepC X \xMerge Y \update_{\alpha \cdot \beta} \Gamma \Delta \sepC  Z \update_\alpha \Lambda \update_{\alpha \cdot \beta'} \Imp B_{k-1}}
            \hypo{}
            \ellipsis{$(J_k)$}{ \cdots \sepC Z \update_\alpha \Lambda \update_{\alpha \cdot \beta'} B_{k-1} \Imp}
        \infer2[(Cut)]{G \sepC G \sepC H \sepC H \sepC X \xMerge Y \update_{\alpha \cdot \beta} \Gamma \Delta \sepC X \xMerge Y \update_{\alpha \cdot \beta}\  \Gamma \Delta \sep Z \xMerge Z \update_\alpha \Lambda \Lambda \update_{\alpha \cdot \beta'} \Imp }
        \infer1[(Merge)$^\ast+$(C)$^\ast$]{G \sepC H \sepC X \xMerge Y \update_{\alpha \cdot \beta} \Gamma \Delta \sep Z \update_\alpha \Lambda \ \update_{\alpha\cdot \beta'} \Imp }
        \infer1[(DW)]{G \sepC H \sepC X \xMerge Y \update_{\alpha \cdot \beta} \Gamma \Delta \sep Z \update_\alpha \Lambda \update_{\alpha \cdot B_1 \cdots B_{k-3}} B_{k-2}\Imp \update_{\alpha\cdot \beta'} \Imp }
        \infer1[(New$^A$)]{(J_{k-1}): \quad G \sepC H \sepC X \xMerge Y \update_{\alpha \cdot \beta} \Gamma \Delta \sep Z \update_\alpha \Lambda \update_{\alpha \cdot B_1 \cdots B_{k-3}} B_{k-2}\Imp }
    \end{prooftree}
    \end{center}

    \hfill\\
    
    \noindent We proceed this way until we obtain $(J_2): G \sep H \sep X \xMerge Y \update_{\alpha \cdot \beta} \Gamma \Delta \sep Z \update_\alpha B_1,  \Lambda$. Then, we use it together with $(I_1)$ to apply (Cut), which application is admissible by induction on the DHS-complexity of the cut-formula. We conclude with the admissibility of the structural rules:
    
    \begin{center}
    \begin{prooftree}[regular]
            \hypo{}
            \ellipsis{$(I_1)$}{G \sep H \sep X \xMerge Y \update_{\alpha \cdot \beta} \Gamma \Delta \sep  Z \update_\alpha  \Lambda,  B_1}
            \hypo{}
            \ellipsis{$(J_2)$}{G \sep H \sep X \xMerge Y \update_{\alpha \cdot \beta} \Gamma \Delta \sep Z \update_\alpha B_1,  \Lambda}
        \infer2[(Cut)]{G \sepC G \sepC H \sepC H \sep X \xMerge Y \update_{\alpha \cdot \beta} \Gamma \Delta \sep X \xMerge Y \update_{\alpha \cdot \beta}\  \Gamma \Delta \sep Z \xMerge Z \update_\alpha \Lambda \Lambda}
        \infer1[(Merge)$^\ast+$(C)$^\ast$]{G \sep H \sep X \xMerge Y \update_{\alpha \cdot \beta} \Gamma \Delta \sep Z \update_\alpha \Lambda }
    \end{prooftree}\\
    \end{center}

    \hfill \\

\end{document}